\documentclass[numberwithinsect]{article}

\usepackage{arxiv}

\usepackage{multirow}
\usepackage{amsmath,amssymb,amsfonts}
\usepackage[utf8]{inputenc} 
\usepackage[T1]{fontenc}    
\usepackage{hyperref}       
\usepackage{url}            
\usepackage{tabularx}
\usepackage{booktabs}       
\usepackage{amsfonts}       
\usepackage{nicefrac}       
\usepackage{microtype}      
\usepackage{lipsum}		
\usepackage{graphicx}
\usepackage[square,numbers]{natbib}
\usepackage{doi}
\usepackage{caption}
\usepackage{subfigure}
\usepackage{listings}
\usepackage{enumitem}

\title{Clock Synchronization in Virtualized Distributed Real-Time Systems using IEEE\,802.1AS and ACRN}

\author{
	\href{https://orcid.org/0000-0002-0305-7139}{\includegraphics[scale=0.06]{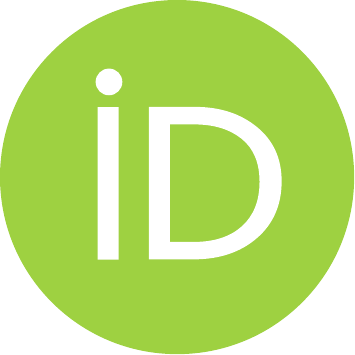}\hspace{1mm}Jan Ruh} \\
	TTTech Computertechnik AG\\
	Vienna, Austria\\
	\texttt{jan.ruh@tttech.com} \\
	\And
	\href{https://orcid.org/0000-0003-0652-5943}{\includegraphics[scale=0.06]{orcid.pdf}\hspace{1mm}Wilfried Steiner} \\
	TTTech Computertechnik AG\\
	Vienna, Austria\\
	\texttt{wilfried.steiner@tttech.com} \\
	\And
	\href{https://orcid.org/0000-0001-6162-2653}{\includegraphics[scale=0.06]{orcid.pdf}\hspace{1mm}Gerhard Fohler} \\
	TU Kaiserslautern\\
	Kaiserslautern, Germany\\
	\texttt{fohler@eit.uni-kl.de} \\
}

\renewcommand{\shorttitle}{\textit{arXiv} Clock Synchronization in Virtualized Distributed Real-Time Systems}
\hypersetup{
pdftitle={Clock Synchronization in Virtualized Distributed Real-Time Systems using IEEE\,802.1AS and ACRN},
pdfsubject={Clock Synchronization, Virtualization, Distributed Systems},
pdfauthor={Jan Ruh, Wilfried Steiner, Gerhard Fohler},
pdfkeywords={Clock Synchronization, Virtualization, IEEE~802.1AS, gPTP, ACRN, Hypervisor, Formalization},
}

\begin{document}
	
\newtheorem{property}{Property}
\newtheorem{definition}{Definition}
\newtheorem{proof}{Proof}
\newtheorem{theorem}{Theorem}
\numberwithin{equation}{section}
\renewcommand{\theequation}{\arabic{section}.\arabic{equation}}

\maketitle

\begin{abstract}
Virtualization of distributed real-time systems enables the consolidation of mixed-criticality functions on a shared hardware platform, thus easing system integration. Time-triggered communication and computation can act as an enabler of safe hard real-time systems. A \emph{time-triggered hypervisor} that activates virtual CPUs according to a global schedule can provide the means to allow for a resource-efficient implementation of the time-triggered paradigm in virtualized distributed real-time systems. A prerequisite of time-triggered virtualization for hard real-time systems is providing access to a global time base to VMs and the hypervisor. A global time base results from clock synchronization with an upper bound on the \emph{clock synchronization precision}.

We present a formalization of the notion of time in virtualized distributed real-time systems. We use this formalization to propose a \emph{virtual clock condition} that enables us to test the suitability of a virtual clock for the design of virtualized time-triggered real-time systems focusing on clock synchronization. We discuss and model how virtualization, particularly resource consolidation versus resource partitioning, degrades clock synchronization precision. Finally, we apply our insights to model the IEEE~802.1AS clock synchronization protocol and derive an upper bound on the clock synchronization precision of IEEE~802.1AS in a virtualized distributed real-time system. We present our implementation of a \emph{dependent clock} for ACRN that can be synchronized to a grandmaster clock. The results of our experiments illustrate that a type-1 hypervisor like ACRN implementing the dependent clock paradigm yields native clock synchronization precision. Furthermore, we show that the upper bound of clock synchronization precision derived from our model holds for a series of experiments featuring native and virtualized setups.
\end{abstract}

\section{Introduction}
The consolidation of mixed-criticality functions on multi-core processor (MCP) hardware platforms has become of significant interest in many industries with applications demanding hard real-time requirements, such as aerospace, industrial automation, and automotive.

Virtualization can achieve strict spatial and temporal isolation of co-located mixed-critical functionality. It guarantees isolation of multiple partitions or virtual machines (VMs) that share a hardware platform by introducing an additional software layer, the so-called hypervisor. Existing real-time hypervisors come with different properties depending on their requirements with respect to, e.g., real-time capabilities, safety, and configurability. We differentiate resource consolidating hypervisors that target more flexible soft real-time systems, such as RT-Xen~\cite{DBLP:conf/emsoft/XiWLG11,DBLP:conf/emsoft/XiXLPGSL14}, and resource partitioning hypervisors that target safety-critical, static, hard real-time systems, such as ACRN~\cite{DBLP:conf/vee/LiXRD19}, Xtratum~\cite{DBLP:conf/edcc/CrespoRM10}, Jailhouse\footnote{\url{https://github.com/siemens/jailhouse}}, QNX\footnote{\url{https://blackberry.qnx.com/en/software-solutions/embedded-software/industrial/qnx-hypervisor}} or PikeOS\footnote{\url{https://www.sysgo.com/products/pikeos-hypervisor/}}, to name only a few. Both approaches can utilize mechanisms to mitigate timing interference due to \emph{implicitly shared resources} such as last level caches (LLCs)~\cite{DBLP:conf/rtas/XuTPCL17} or the main memory bus~\cite{DBLP:conf/rtas/YunYPCS13}. Furthermore, resource consolidating hypervisors might apply theoretical insights from hierarchical scheduling~\cite{DBLP:conf/ecrts/DengLS97,DBLP:conf/rtss/KuoL99,DBLP:conf/rtss/ShinL03,DBLP:conf/rtss/DavisB05,DBLP:conf/emsoft/XiWLG11,DBLP:conf/emsoft/XiXLPGSL14} and response time analysis~\cite{DBLP:conf/rtcsa/BalbastreRC09,DBLP:conf/emsoft/AlmeidaP04,DBLP:conf/dac/BeckertNEP14} to give timing guarantees to VMs that share processor cores.

However, both resource consolidation and partitioning come with drawbacks and trade-offs. Resource consolidating hypervisors have to trade-off their real-time capabilities for resource utilization. Partitioning hypervisors lack scalability since they are fundamentally bound by the available physically partitionable hardware, such as CPU cores or network interface cards (NICs). Temporal partitioning of a CPU core by applying time division multiple access (TDMA) or time-triggered scheduling is often applied in safety-critical domains. It allows sharing a processor core~\cite{DBLP:conf/dac/BeckertNEP14,DBLP:conf/dac/BeckertE15,DBLP:conf/edcc/CrespoRM10} or a communication channel~\cite{obermaisser2011} while providing deterministic timing. Not a lot of attention has been paid to problems arising from the interplay of virtualized real-time systems, time-triggered communication (i.e., time-triggered messages), and time-triggered computation (i.e., time-triggered tasks). A central problem, e.g., in automotive domain controllers, guarantees the dispatching of time-triggered tasks according to a global communication and computation schedule. If time-triggered tasks are located in a VM, this implies that the VM must run at the exact point in time a time-triggered task is to be dispatched. We acknowledge the possibility of dynamic non-time-triggered scheduling of a VM's virtual CPUs (vCPUs) with dynamic or static priorities so that time-triggered tasks inside the VMs can follow a static dispatch schedule, e.g., by deriving suitable VM deadlines from the static combined task and network schedule. However, this results in a complex hierarchical scheduling problem~\cite{DBLP:conf/rtss/ShinL03} and makes it hard to argue about the timing properties of a hypervisor and its VMs, e.g., for worst-case response time, especially on multi-core platforms. Therefore, we propose the concept of a \emph{time-triggered hypervisor} that dispatches vCPUs according to a global static schedule, thus preserving the timing properties of time-triggered tasks and communication.

A prerequisite of time-triggered communication and computation is that all nodes of a distributed real-time system share a time base. We presume a distributed real-time system in which every node has access to a local physical clock that consists of an oscillator and a counter that increments with the oscillator's frequency. The oscillator comes with a nominal frequency that gives the granularity of the clock that is the least measurable time interval. Even if given a distributed real-time system that consists of homogeneous nodes that all have been started at the same time and run with the same nominal frequency, their local physical clocks will drift apart due to factors such as manufacturing variations, temperature, or aging of the oscillator~\cite{DBLP:conf/sigmetrics/SchmidCFCS08}. Therefore, the nodes execute a clock synchronization protocol that allows them to synchronize their local physical clocks to each other by exchanging local time information. As a result, they can establish a global time base with an a priori known \emph{clock synchronization precision}. A typical value for the precision of clock synchronization that allows for the execution of distributed control algorithms is $1\,\mu s$ or lower. The IEEE\,802.1AS clock synchronization standard also advertises a precision of below one microsecond for networks with up to six hops~\cite{ieee8021as}. The introduction of virtualization into distributed real-time systems to create logical partitions, e.g., partitioning domain controllers to consolidate virtual ECUs, raises the question of how to synchronize the clocks of VMs achieving a sub-microsecond clock synchronization precision and how to synchronize the hypervisor's local physical clock so it can dispatch vCPUs according to a global schedule.

We identify two dimensions to this problem. Firstly, the technical challenge is implementing clock synchronization in a virtualized system and achieving sub-microsecond clock synchronization precision. For example, Broomhead~et~al.~\cite{DBLP:conf/osdi/BroomheadCRV10} identify overhead associated with running one instance of the Network Time Protocol~(NTP) per VM and refer to this as the \emph{independent clock paradigm}. They propose an alternative implementation based on a \emph{feed-forward} mechanism that resorts to a \emph{dependent clock paradigm}. A single VM executes NTP and shares the synchronized time with co-located VMs. The authors of QuartzV~\cite{DBLP:conf/rtas/DSouzaR18} pick up on the dependent clock paradigm for virtualizing real-time systems and introduce the concept of \emph{timelines} whereas each timeline provides a different quality of time to a VM. The timelines' clock synchronization is based on a Linux implementation~\cite{cochran:2010,cochran:2011} of the IEEE\,1588-2008~\cite{ieee1588-2008} clock synchronization standard. In a para-virtualized setup of their implementation of the dependent clock paradigm for Linux KVM in which the host OS performs clock synchronization, they achieve an accuracy of $26.12 \pm 5.12\mu s$. In a fully virtualized setup, in which a VM executes the clock synchronization protocol, they measure an accuracy of $70.23\pm 128.28 \mu s$. These reported results are significantly higher than the precision of less than one microsecond that is a common magnitude for this parameter in industry. Furthermore, the existing work~\cite{DBLP:conf/osdi/BroomheadCRV10,DBLP:conf/rtas/DSouzaR18} does not consider issues that might arise if synchronizing the local physical clock of the hypervisor to other physical clocks in a distributed real-time system with the intent to use it to implement time-triggered scheduling of vCPUs.

The second dimension of the problem is a lack of a rigorous formalization of clock synchronization in virtualized environments. Formalization enables analyzing the properties of clock synchronization and deriving a verifiable upper bound for the achievable clock synchronization precision. Formalization is crucial for a time-triggered hypervisor that shall be used in safety-critical components of a time-triggered real-time system since the precision of clock synchronization is an essential parameter for time-triggered schedule synthesis~\cite{DBLP:conf/rtss/Steiner10,DBLP:conf/isorc/Steiner11,DBLP:journals/rts/CraciunasO16}.

The contributions of this work are:
\begin{itemize}
	\item We propose three properties that a virtual clock must provide to fulfill a \emph{virtual clock condition}. The virtual clock condition is a necessary but not sufficient condition for a virtual clock to be considered a \emph{good clock} according to Kopetz~\cite{DBLP:journals/tc/KopetzO87,DBLP:books/sp/Kopetz11}.
	\item We provide a thorough analysis of clock synchronization in virtualized distributed real-time systems identifying how hypervisor latency degrades clock synchronization precision.
	\item We discuss the dependent clock paradigm and, based on our analysis of clock synchronization in virtualized systems, argue why we expect it to yield near-native clock synchronization precision. Furthermore, we provide an implementation of a dependent clock for the ACRN hypervisor.
	\item We utilize our formalization of clock synchronization in virtualized distributed real-time systems to derive an upper bound for the clock synchronization precision of IEEE\,802.1AS~\cite{ieee8021as}.
	\item We present an experimental evaluation of our dependent clock implementation for ACRN and the upper bound on clock synchronization precision of IEEE\,802.1AS. To that end, we describe an elaborate test setup featuring three commercial edge computing devices. Furthermore, we can use our insights of how virtualization degrades clock synchronization precision to identify and account for hypervisor latency when measuring clock synchronization precision in virtualized systems.
\end{itemize}

The remainder of the paper is structured as follows: In Section~\ref{sec:notion_of_time}, we propose a formalization of the notion of time and clock synchronization in virtualized systems that follows the work of Kopetz~\cite{DBLP:journals/tc/KopetzO87,DBLP:books/sp/Kopetz11}. We propose the virtual clock condition and introduce the concepts of \emph{continuous} and \emph{discontinuous} virtual clocks illustrating fundamental constraints when virtualizing real-time systems. In Section~\ref{sec:clock_sync_in_virtualized_distributed_rt_systems}, we illustrate how virtualization degrades the precision of clock synchronization in distributed real-time systems. We continue in Section~\ref{sec:analysis_ieee802.1as_virtualized_distributed_system} by analyzing clock synchronization in virtualized distributed real-time systems deriving an upper bound on the achievable clock synchronization precision of IEEE\,802.1AS~\cite{ieee8021as}. Finally, in Section~\ref{sec:evaluation}, we present our dependent clock implementation for the ACRN hypervisor and show the applicability of our formal model to real-life systems. To that end, we compare the clock synchronization precision derived from our model to the clock synchronization precision measured in an experimental setup of a virtualized distributed real-time system for various hypervisor configurations ranging from resource consolidation to resource partitioning.

\section{Notion of Time in Virtualized Distributed Real-Time Systems}
\label{sec:notion_of_time}
We base our model of time in virtualized distributed real-time systems on the work of Kopetz~\cite{DBLP:journals/tc/KopetzO87,DBLP:books/sp/Kopetz11} and follow along with his definitions in order to transfer his insights on clock synchronization to virtualized distributed real-time systems\footnote{A summary of the used notation can be found in Appendix~\ref{app:notation_summary}.}. Kopetz uses the concepts of an \textit{event} and a \textit{duration} to keep track of time. An event is a happening at a definite point in time that itself does not take any time, whereas a duration is the time between two events. Furthermore, we presume the existence of a reference clock\footnote{Note, that such a reference clock does not exist -- it rather is a theoretical construct that aids the construction of a clock model.} $C^r$ with an undefined, yet sufficiently small granularity \textit{$g^r$}, thus minimizing the discreteness of the time base. We presume that the reference clock is in perfect synchronization with the international atomic time TAI~\cite{bipm2018}. This means that reading the reference clock at any point in time yields exactly the state of the TAI at that point in time. The microticks of the reference clock represent the smallest notion of the passage of time, so the observation of an event $e$ can be assigned a timestamp $ts(e)$ that is the state of the reference clock at its latest microtick. As a result, two or more events following each other between two consecutive microticks of the reference clock appear to occur at the same point in time.

We presume an encapsulated distributed system with a finite number of physical nodes $|N|>1$. Each node $k\in \mathcal{N}$ comes with a \textit{local physical clock} $C^k$ providing a \textit{local time} $tl^k$. We denote the event of a microtick $l\in\mathbb{N}_0$ of $C^k$ with $tl_l^k$ and the timestamp of an event $e$ in the local time with $tl^k(e)=l$. The event's timestamp is given by the state of the local physical clock at its latest microtick $l$. We note that a microtick is always merely a counter value without a unit. The time unit always comes in when combining the counter value with the granularity of the clock. The granularity $g^k$ of a local clock is given by the nominal number of microticks of the reference clock between two consecutive microticks of the local clock. Kopetz defines the drift $r^k_{l,l+1}$ of a physical clock as the ratio of the number of microticks of the reference clock between two consecutive microticks $l$ and $l+1$ of a local physical clock and its nominal granularity $g^k$:
\begin{align}
r^{k}_{l,l+1}=\frac{|ts(tl^{k}_{l+1})-ts(tl^{k}_{l})|}{g^{k}}
\end{align}
Furthermore, he defines a \emph{good clock} as a clock $C^k$ whose drift~$r_{l,l+1}$ is bounded by a maximum drift rate $r_{max}$.
\begin{definition}
	\label{def:max_drift_rate}
	A local physical clock $C^k$ of a node $k$ is a good clock if:
	\begin{align}
	\forall k \text{ and }\forall l\in\mathbb{N}_0: 1-r_{max}\le r^k_{l,l+1}\le 1+r_{max}
	\end{align}
\end{definition}
An upper bound on the drift rate of all clocks in a distributed real-time system is a precondition necessary in order to provide guarantees for the precision of clock synchronization.

\subsection{Virtual Clocks}
\label{ssec:virtual_clocks}
In addition to the work of Kopetz~\cite{DBLP:journals/tc/KopetzO87,DBLP:books/sp/Kopetz11}, we also consider that there can be a finite number $|\mathcal{V}^k|\geq0$ of VMs running on each physical node $k\in \mathcal{N}$. We refer to a node that is hosting VMs as a host node. A VM has access to a set of virtual resources mapped to physical resources by the \emph{hypervisor} running on the host node. Note that this way, our analysis of clock synchronization can cover both resource consolidating and resource partitioning hypervisors. We can model resource partitioning as one-to-one virtual to physical resource mappings, making it a special case of resource consolidation. Having this in mind, we will proceed with analyzing clock synchronization in virtualized distributed real-time systems, presuming a resource consolidating hypervisor. In particular, when saying resource consolidating hypervisor, what we mean is that the hypervisor schedules VMs vCPUs on the available pCPUs, meaning that a VM's execution can get preempted by the hypervisor depending on its implemented scheduling policy. However, resource consolidation can also include virtual networking or the sharing of other I/O devices.

Now, in alignment with the definitions of Kopetz, we introduce a \emph{local virtual clock} $C^{k,j}$ that, given the local physical clock $C^k$, grants a VM $j$ access to a local virtual time $tv^{k,j}$.
\begin{definition}
	A local virtual clock $C^{k,j}$ provides a local virtual time $tv^{k,j}$. We define the microtick of a virtual clock $C^{k,j}$ as a function $\nu^{k,j}:\mathbb{N}_0\rightarrow \mathbb{N}_0$ of the microtick $l$ of the local physical clock $C^k$ of the host node $k$. We refer to the function $\nu^{k,j}$ as the virtual microtick of the local virtual clock $C^{k,j}$ whereas $tv^{k,j}_{\nu(l)}$ denotes the event of the $\nu(l)$-th microtick of virtual clock $C^{k,j}$. The timestamp $tv^{k,j}(e)$ of an event $e$ in the virtual time of a VM $j$ running on a host node $k$ is given by the state of the virtual clock at its latest virtual microtick $\nu^{k,j}$.
\end{definition}
If considering an arbitrary VM on an arbitrary host node we leave out the superscript indices $k,j$ for better readability, e.g., writing $tv(e)=\nu(l)$.

\subsection{The Virtual Clock Condition}
When trying to apply Kopetz definitions of a clock drift rate and a good clock on virtual clocks, we find that not all virtual clocks are suited for a virtualized real-time system. First of all, we notice that the definition of a physical clock's granularity $g^k$ as the nominal number of microticks of the reference clock between two consecutive microticks of the local clock does not apply to an arbitrary virtual clock. For instance, given $n\in\mathbb{N}$ consecutive virtual microticks $\nu(l),\nu(l)+1,...,\nu(l)+n$ it is impossible to determine the nominal number of microticks of the reference clock in between two consecutive virtual microticks due to the possibility of an occurring preemption. Thus the question arises of how to deal with the progression of virtual time during preemption. Related to this is that a time measurement using a virtual clock might yield an error-prone result due to the possibility of invoking the hypervisor during the time measurement period. As a result, we propose two fundamental properties a virtual clock must have to be suited for use in a virtualized time-triggered real-time system. We can show that a virtual clock with said properties can always be considered a good clock according to Kopetz~\cite{DBLP:books/sp/Kopetz11} making it suitable for use in virtualized time-triggered real-time systems.
\begin{property}[Succession Property]
	\label{prop:succession_property}
	A virtual clock $C^{k,j}$ fulfills the succession property if the virtual microtick $\nu(l)$ increases by one for every microtick of the local physical clock $C^k$:
	\begin{align}
	l_0\in\mathbb{N}\text{ and } l\ge l_0: \nu(l+1)=\nu(l)+1 \nonumber
	\end{align}
\end{property}
\begin{property}[Simultaneity Property]
	\label{prop:simultaneity_property}
	A virtual clock $C^{k,j}$ fulfills the simultaneity property if the virtual microtick $\nu(l)$ and the microtick of the local physical clock $C^k$ occur at the exact same time:
	\begin{align}
	l_0\in\mathbb{N}\text{ and } l\ge l_0: ts(tv_{\nu(l)})=ts(tl_l) \nonumber
	\end{align}
\end{property}
\begin{theorem}[Virtual Clock Condition]
	\label{theorem:virtual_clock_condition}
	If a virtual clock $C^{k,j}$ fulfills the succession Property~\ref{prop:succession_property}, the simultaneity Property~\ref{prop:simultaneity_property}, and the local physical clock $C^k$ it is derived from is a good clock according to Definition~\ref{def:max_drift_rate}, then the local virtual clock $C^{k,j}$ itself is a good clock as well.
\end{theorem}
\begin{proof}
	Given that the local physical clock $C^k$ is a good clock, the virtual clock $C^{k,j}$ is a good clock as well if and only if its drift rate $r^{k,j}_{\nu(l),\nu(l+1)}$ is lower or equal the physical clock's drift rate $r^{k}_{l,l+1}$. Therefore, we have to show that for $l_0\in\mathbb{N}\text{ and } l\ge l_0$:
	\begin{align}
	\label{eq:virtual_drift_rate}
	\frac{|ts(tv^{k,j}_{\nu(l+1)})-ts(tv^{k,j}_{\nu(l)})|}{g^{k,j}} \le \frac{|ts(tl^{k}_{l+1})-ts(tl^{k}_{l})|}{g^{k}}
	\end{align}
	The succession property and the simultaneity property imply that every virtual microtick of the virtual clock $C^{k,j}$ corresponds to a microtick of the local physical clock $C^{k}$. As a result, the granularity $g^{k,j}$ of the virtual clock $C^{k,j}$ must be equal to the granularity $g^k$ of the local physical clock $C^k$. If we further successively apply the succession and the simultaneity property to the drift rate of the virtual clock in Equation~(\ref{eq:virtual_drift_rate}), we get:
	\begin{align}
	& & \frac{|ts(tv^{k,j}_{\nu(l+1)})-ts(tv^{k,j}_{\nu(l)})|}{g^{k,j}} & \le \frac{|ts(tl^{k}_{l+1})-ts(tl^{k}_{l})|}{g^{k}} \nonumber \\
	& \stackrel{Prop.\,\ref{prop:succession_property}}{\iff} & \frac{|ts(tv^{k,j}_{\nu(l)+1})-ts(tv^{k,j}_{\nu(l)})|}{g^{k,j}} & \stackrel{Prop.\,\ref{prop:simultaneity_property}}{=} \frac{|ts(tl^{k}_{l+1})-ts(tl^{k}_{l})|}{g^{k}} \nonumber
	\end{align}
	From $r^{k,j}_{\nu(l),\nu(l+1)}=r^{k}_{l,l+1}$ follows directly that the virtual clock $C^{k,j}$ is a good clock.
\end{proof}

For the remainder of this paper, we will focus on a local virtual time derived from the local physical time by applying a \emph{clock offset} to the microtick $l$ of the physical local clock $C^k$. Note that the value of the clock offset can theoretically change with each microtick of the host node's local physical clock. Therefore, we introduce a function $o(l)$ that returns the clock offset of a virtual clock at the microtick $l$ of the host node. The virtual microtick of a local virtual clock is then given by $\nu(l)= l-o(l)$. We want to point out that the value returned by the virtual microtick, just as a physical microtick, is a counter value and not a time value. An operating system~(OS) uses the counter value of a virtual or physical clock and its granularity to derive the actual time.

The clock offset $o(l)$ at a microtick $l$ of the host node depends on implementing the virtual to local clock mapping of a VM in the hypervisor.

\subsection{Continuous and Discontinuous Clocks}
\label{ssec:continuous_discontinuous_clocks}
As mentioned before, in general, we assume that a hypervisor schedules VMs vCPUs on the available pCPUs, meaning that a VM's execution can get preempted by the hypervisor depending on its implemented scheduling policy. We refer to the event of the $x$-th preemption of a VM on a host node as a \emph{preempt} event denoted with $preempt_x$. Accordingly, we refer to resuming the execution of a VM after the $x$-th preemption as a \emph{VM resume} event denoted with $resume_{x}$.

There are two ways to deal with the progression of virtual time during preemption, either the virtual clock keeps ticking during preemption, or the virtual clock is stopped and continues ticking once the VM resumes. In Figure~\ref{fig:preemption_discontclock}, we can see a virtual clock that keeps ticking during preemption causing the VM to miss several virtual microticks. The VM perceives a hole in its virtual time, jumping from the virtual microtick value $\nu(l)=l$ to the value $\nu(l+\delta+1)=l+\delta+1$. We refer to this as a \emph{discontinuous virtual clock}~(DVC) since the VM experiences a discontinuous virtual time. The clock offset $o(l)$ of the virtual microtick $\nu(l)$ of a DVC only accounts for a constant value relating to the difference between the start time of the VM and the host node. This difference is equal to the timestamp $tl(e_{start})$ of a VM's start event $e_{start}$ in the local time of its host node. For the sake of simplicity, we do not show this in Figure~\ref{fig:virtual_drift_preemption}.

\begin{figure}
	\centering
	\subfigure[Clock keeps ticking during preemption.]{
		\includegraphics[width=0.9\textwidth,trim={0.7cm 0.1cm 0.2cm 0.1cm},clip]{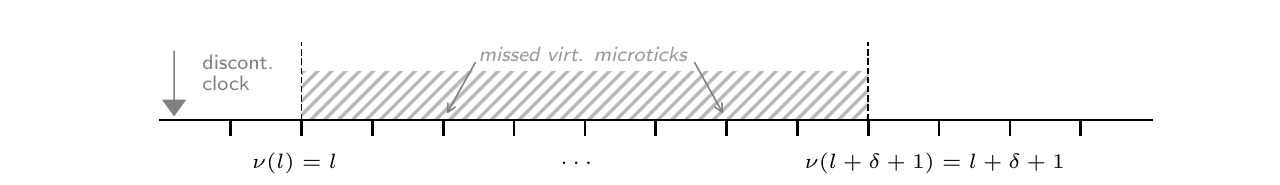}
		\label{fig:preemption_discontclock}
	}
	\subfigure[Clock is stopped during preemption.]{
		\includegraphics[width=0.9\textwidth,trim={0.7cm 0.1cm 0.2cm 0.1cm},clip]{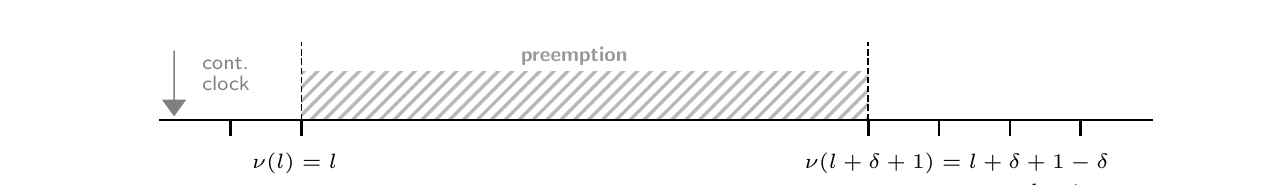}
		\label{fig:preemption_contclock}
	}
	\subfigure[Physical Local Clock and Reference Clock.]{
		\includegraphics[width=0.9\textwidth,trim={0.7cm 0.4cm 0.2cm 0.25cm},clip]{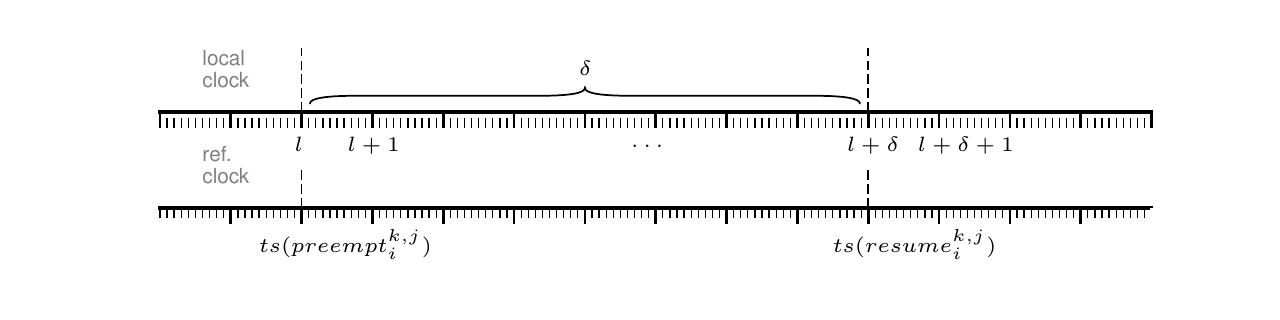}
		\label{fig:preemption_localclock}
	}
	\caption{Comparing continuous and discontinuous virtual clocks in case of preemption. In (a) we show a virtual clock that keeps ticking during preemption thus missing virtual microticks yielding a discontinuous virtual time. In (b) on the other hand we show a virtual clock that is stopped during preemption yielding a continuous virtual time.}
	\label{fig:virtual_drift_preemption}
\end{figure}

In contrast to Figure~\ref{fig:preemption_discontclock}, the virtual clock is stopped during preemption in Figure~\ref{fig:preemption_contclock}. Therefore, the VM perceives a continuous virtual time provided by a \emph{continuous virtual clock} (CVC). The CVC accounts for the duration of preemptions so that the virtual time progresses without jumps when the hypervisor resumes the VM. Let $P_l$ be a VM's set of preemptions until microtick $l$ of its host node. Furthermore, $p_i\in P_l$ denotes the duration of the $i$-th preemption of a VM in microticks of its host node's local physical clock. We can derive the duration of the $i$-th preemption the difference between the timestamps of the respective VM $preempt$ and VM $resume$ events in the time of the local physical clock, as shown in Figure~\ref{fig:preemption_localclock}:
\begin{align}
\delta:=p_i=tl(resume_i)-tl(preempt_i)
\end{align}
We add up the duration of all preemptions to obtain a value for the clock offset $o(l)$ of the CVC up to the $l$-th microtick of its host node. This leads us to the following distinction of cases for the clock offset of DVCs and CVCs.
\begin{definition}
	The clock offset $o(l)$ of a virtual clock $C^{k,j}$ from the local physical clock $C^k$ of its host node at microtick $l$ is given by:
	\begin{align}
	\label{eq:virtual_clock_offset}
	o(l)=\begin{cases}tl(e_{start}), & \text{discontinuous clock} \\ tl(e_{start}) + \sum^{|P_l|-1}_{i=0}p_i, & \text{continuous clock}\end{cases}
	\end{align}
\end{definition}

With these specific definitions of two virtual clock implementations that differ in how they handle preemptions by the hypervisor, we will check their suitability for use in virtualized real-time systems by testing for the succession and simultaneity properties. We consider a virtual clock a good clock if it fulfills both properties.

In contrast to what the name suggests, a DVC indeed does fulfill the succession property since for $l_0 = tl(e_{start}), l\ge l_0$ it is:
\begin{align}
\nu(l+1)=l-l_0+1=\nu(l)+1
\end{align}
The mismatch between the naming of the DVC and its properties results from whether we consider the progression of virtual time $tv$ or local physical time $tl$. From a VM's point of view, the progression of time is discontinuous since the VM perceives a hole in virtual time. However, from the host's point of view, the progression of virtual time follows exactly the progression of local physical time, just with an offset of $l_0$. This is illustrated in Figure~\ref{fig:discontinuous_virtual_time}. For the simultaneity property to hold, we have to make assumptions on implementing the virtual to physical time mapping. Therefore, we assume that adding the offset $o(l)$ to the microtick $l$ is performed when the VM reads the virtual clock, which conforms with the implementation in modern processor architectures, e.g., the timestamp counter (TSC) in the x86 architecture\cite{intel2020}. As a result, the virtual microtick and the microtick of the host's physical clock are the same events entailing the simultaneity property. Assuming that the local physical clock $C^k$ is a good clock, this implies that any DVC $C^{k,j}$ based on that local physical clock is a good clock as well.

\begin{figure*}
	\subfigure[Progression of discontinuous virtual time.]{
		\includegraphics[width=0.45\textwidth,trim={0.0cm 1.6cm 0.0cm 0.5cm},clip]{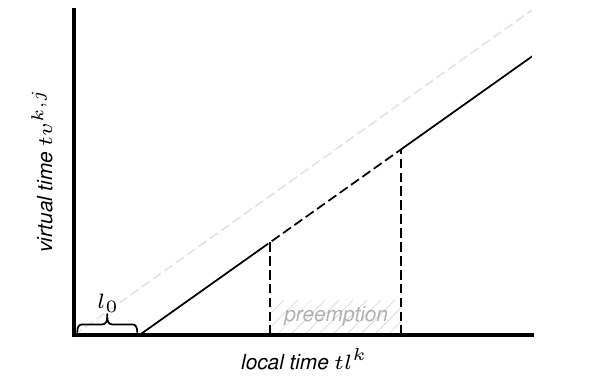}
		\label{fig:discontinuous_virtual_time}
	}
	\subfigure[Progression of continuous virtual time.]{
		\includegraphics[width=0.45\textwidth,trim={0.0cm 1.6cm 0.0cm 0.5cm},clip]{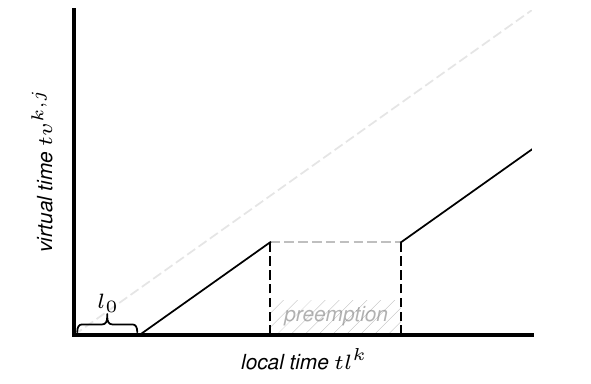}
		\label{fig:continuous_virtual_time}
	}
	\caption{We are comparing the progression of virtual time in case of a discontinuous versus a continuous virtual clock.}
	\label{fig:example}
\end{figure*}
For a CVC, we find that the succession property does not hold in case of a preemption between a virtual microtick $\nu(l)$ and $\nu(l)+1$. In case of a preemption by the host at microtick $l+1$, the set of preemptions $P_{l+1}$ is in an undefined state since it is about to be updated by the duration $p_{l+1}>0$ of the ongoing preemption. For the succession property this yields:
\begin{align}
\nu(l+1)-\nu(l)= 1-\big(\sum^{|P_{l+1}|-1}_{i=0}p_i\big)+\sum^{|P_l|-1}_{i=0}p_i
\end{align}
Note that the difference between the two preemption sums can be substituted by the preemption $p_{l+1}$ that is in progress which results in a violation of the succession property:
\begin{align}
\nu(l+1)-\nu(l)=1-p_{l+1} < 1
\end{align}
As a result, we can conclude that CVCs cannot be considered good clocks according to the virtual clock condition.

Even though the rationale backing these results might seem trivial in retrospect, in our experience, these fundamental insights are not very widespread in industry. Furthermore, the mathematical understanding that CVCs can never be good clocks has clear implications for the design of virtualized real-time systems. Firstly, in the case of a partitioning hypervisor with a one-to-one mapping of vCPUs to pCPUs, there is an equivalency of CVCs and DVCs due to the absence of preemptions. Therefore, as long as the number of distinct functions that shall be consolidated on a single platform is less or equal to the number of pCPUs, a partitioning approach featuring one-to-one mapping of functions to pCPUs is feasible. However, e.g., in the automotive industry, the number of virtual ECUs exceeds the number of pCPUs rendering a one-to-one mapping impossible, calling for the consolidation of mixed-critical functionality on a shared pCPU. We have established that the preemptions inherent with vCPU scheduling require a DVC. However, the use of a DVC implies that it is impossible to design virtualized time-triggered real-time systems that consolidate functionality on a processor core and are unaware of their virtualization. For example, it is not feasible to migrate legacy software, which expects a continuous progression of time, to execute inside a VM on a consolidating hypervisor without porting overhead, making the software aware of the holes in the timeline.

For the remainder of the paper, if not stated differently, a virtual clock always refers to a DVC, and we refer to a VM with a good clock as a \emph{virtual node}.

\subsection{Clock Synchronization Precision}
\label{ssec:clock_sync_precision}
With the definition of virtual nodes, we have to extend our model of a distributed real-time system. Therefore, we summarize all physical nodes $k\in \mathcal{N}$ and virtual nodes $j\in \mathcal{V}^k$ in a set:
\begin{align}
D:=\mathcal{N}\cup\bigcup_{k\in \mathcal{N}}\mathcal{V}^k
\end{align}
Each node $n\in D$ has a finite number of neighboring nodes $\mathcal{N}^n\subset D$. The timestamp of an event $e$ of a physical or virtual node $n\in D$ is given by $tn^n(e)$, and the microtick of a good clock $C^n$ is denoted by $i\in\mathbb{N}_0$. Furthermore, we assume all nodes $n\in D$ come with a good clock. Let
\begin{align}
G:=\{C^k,C^{k,j}|k\in \mathcal{N}\wedge\forall k: j\in \mathcal{V}^k\}
\end{align}
be the set of all good clocks. Given a set of good clocks that are synchronized via an arbitrary clock synchronization protocol.
\begin{definition}
	\label{def:clock_synchronization_precision}
	The precision of clock synchronization $\Pi$ is given by the maximum difference between any two good clocks in the distributed real-time system $D$:
	\begin{align}
	\Pi:=\max_{\forall n,n'\in D,\forall i\in \mathbb{N}_0}(|ts(tn^n_i)-ts(tn^{n'}_i)|)
	\end{align}
\end{definition}
For the remainder of the paper, we will discuss how virtualization degrades the clock synchronization precision, and we derive an upper bound on clock synchronization precision for virtualized distributed real-time systems, specifically for the IEEE\,802.1AS~\cite{ieee8021as} clock synchronization standard.

\section{Clock Synchronization in Virtualized Distributed Real-Time Systems}
\label{sec:clock_sync_in_virtualized_distributed_rt_systems}

The nodes of a virtualized distributed real-time system $D$ establish a shared time base by exchanging \emph{clock synchronization messages} at a fixed period $S$ that we call the \emph{resynchronization period}. We refer to a specific resynchronization period with $s\in\mathbb{N}_0$.
\begin{definition}
	\label{def:sync_message}
	A clock synchronization message from a sender $n\in D$ to a receiver $n'\in D$ is a tuple containing an estimated message delay $e^{n,n'}_s$ accounting for the end-to-end latency between nodes, a transmission timestamp $tn^{n}(tx_s^{n'})$ of the message on node $n$, and an arbitrary payload $\rho_s$ that depends on the specific clock synchronization protocol.
	\begin{align}
	m_s^{n,n'}:=(e^{n,n'}_s,tn^{n}(tx_s^{n'}),\rho_s)
	\end{align}
\end{definition}
We denote the event of receiving a message $m_s^{n,n'}$ with $rx_s^n$. The receiver can combine its reception timestamp $tn^{n'}(rx_s^n)$ and the information in the clock synchronization message to derive the time offset between its own and the sender's clock.

\subsection{The Reading Error}
\label{ssec:reading_error}
We provide a brief overview of Kopetz work on clock synchronization in distributed real-time systems~\cite{DBLP:journals/tc/KopetzO87,DBLP:books/sp/Kopetz11}. For more details, we refer directly to his work. Kopetz calls the time offset between a receiver's and a sender's clock the \emph{correction term}.
\begin{definition}[Correction Term~\cite{DBLP:journals/tc/KopetzO87}]
	\label{def:correction_term}
	The correction term $c^{n,n'}_s$ equals the difference of the time of reception of message $tn^{n'}(rx_s^n)$, the time of transmission $tn^{n}(tx_s^{n'})$, and the estimated message delay $e^{n,n'}_{s}$ that deviates from the actual message delay $a^{n,n'}_{s}$ by the reading delay~$\epsilon_s^{n,n'}$:
	\begin{align}
	c^{n,n'}_s = tn^{n'}(rx_s^n) - tn^{n}(tx_s^{n'}) - e^{n,n'}_{s} = \Delta_s^{n,n'} + \epsilon_s^{n,n'}\nonumber
	\end{align}
	Therefore, the correction term breaks down into the the actual time difference $\Delta_s^{n,n'}$ between clock $C^n$ and clock $C^{n'}$ and the reading delay $\epsilon_s^{n,n'}=a^{n,n'}_{s}-e^{n,n'}_{s}$.
\end{definition}

\begin{figure}[!t]
	\centering
	\includegraphics[width=5in,trim={0.75cm 0.0cm 0.5cm 0.05cm},clip]{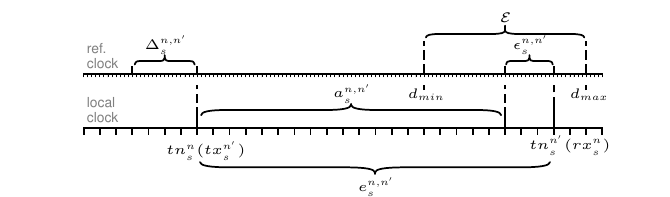}
	\caption{Showing the relationship between the reading delay $\epsilon_s^{n,n'}$, the actual propagation delay $a^{n,n'}_s$, the estimated delay $e^{n,n'}_s$, and the reading error $\mathcal{E}$. The reading delay is bound by the minimum $d_{min}$ and maximum $d_{max}$ message delay.}
	\label{fig:reading_error}
\end{figure}
The reading delay $\epsilon_s^{n,n'}$ between two nodes $n$ and $n'$ represents the discrepancy between the actual time $a^{n,n'}_s$ it takes a synchronization message from node $n$ to node $n'$ and the estimated delay $e^{n,n'}_{s}$. This discrepancy traces back to variable latencies on a communication path from a sender $n$ of a clock synchronization message to a receiver $n'$. These discrepancies are not compensated for and are typically caused by jitter such as the send time jitter, the medium access time jitter, the jitter of the propagation delay on the medium, and the receive time jitter~\cite{DBLP:journals/tc/KopetzO87,DBLP:books/sp/Kopetz11}. The sum of all variable latencies on all possible communication paths from a sender to a receiver of a clock synchronization message in a distributed real-time system is bound by the minimum message delay $d_{min}$ and the maximum message delay $d_{max}$. As a result, for any two reading delays $\epsilon_s^{n,n'}$ and $\epsilon_s^{n,n^*}$ during a resynchronization period $s$ it always is:
\begin{align}
\label{eq:min_max_message_delay}
|\epsilon_s^{n,n'}-\epsilon_s^{n,n^*}|\le d_{max}-d_{min}
\end{align}
We call the difference between the maximum $d_{max}$ and the minimum $d_{min}$ message delay in a distributed real-time system the \emph{latency jitter} or the \emph{reading error}.
\begin{definition}
	\label{def:reading_error}
	The reading error $\mathcal{E}$ is given by the difference of the maximum and minimum possible delay of a synchronization message in the virtualized distributed system $D$:
	\begin{align}
	\mathcal{E}=d_{max}-d_{min} \nonumber
	\end{align}
	It provides an upper bound for the reading delay $\epsilon^{n,n'}_s$ between two nodes $n$ and $n'$ during any resynchronization period $s$.
\end{definition}
The relation between reading delay and reading error is illustrated in~Figure~\ref{fig:reading_error}.

\subsection{Synchronization Condition}
\label{ssec:synchronization_condition}
We can see that the reading error bounds the difference between the clocks of two nodes $n$ and $n'$ since in the worst-case it is $c^{n,n'}=\Delta_s^{n,n'} + \mathcal{E}$. As such, we can utilize the reading error in combination with the maximum clock drift rate $r_{max}$ (c.f.~Definition~\ref{def:max_drift_rate}) to derive an upper bound on clock synchronization precision as follows. After synchronization, two good, free running clocks can drift apart for the time of the synchronization period $S$ with a rate of at maximum $2\cdot r_{max}$ yielding the \emph{drift offset} $\Gamma=2\cdot r_{max}\cdot S$. The reading error $\mathcal{E}$, the drift offset $\Gamma$, and the precision $\Pi$ of clock synchronization yield the \emph{clock synchronization condition}~\cite{DBLP:books/sp/Kopetz11}.
\begin{theorem}[Clock Synchronization Condition~\cite{DBLP:books/sp/Kopetz11}]
	\label{theorem:clock_synchronization_condition}
	Given the reading error $\mathcal{E}$ and the drift offset $\Gamma$, a virtualized distributed real-time system can be synchronized with precision $\Pi$ only if following inequality holds:
	\begin{align}
	\Gamma + \mathcal{E} \leq \Pi
	\end{align}
\end{theorem}

\subsection{Degradation of Clock Synchronization Precision}
\label{ssec:degradation_of_clock_sync_precision}

The clock synchronization condition illustrates that the achievable clock synchronization precision degrades with the drift offset and the reading error. Now, we discuss how clock synchronization precision is affected by the drift offset in a virtualized distributed real-time system. The effect of the drift offset on clock synchronization precision is straightforward since it is determined by known parameters, namely, by the maximum drift rate of the oscillators of the local clocks and the length of the clock synchronization period. The introduction of virtual clocks does not negatively impact the maximum drift rate encountered in a distributed system as long as they fulfill the virtual clock condition, which we have shown to be true for DVCs in Section~\ref{ssec:continuous_discontinuous_clocks}. As a result, there is no degradation of clock synchronization precision concerning the drift offset.

Determining the effect of virtualization on the reading error, on the other hand, requires us to have a closer look at the sources of latency and its jitter on a communication path between the sender and the receiver of a synchronization message. To that end, we state a formula to calculate the minimum $d^{k,k'}_{min}$ and maximum $d^{k,k'}_{max}$ message delays and with them the reading error by identifying all sources of latency between two directly connected physical nodes $k$ and $k'$. In a second step, we will extend both formulas to account for the introduction of virtualization to nodes $k$ and $k'$. Finally, we will expand the minimum, and maximum message delay to a path in a virtualized distributed real-time system~$D$ and use the resulting formulas to derive the reading error.

The message delay between two neighboring nodes $k$ and $k'$ starts accumulating with the creation of the transmission timestamp on node $k$ that can already be afflicted by high jitter depending on whether timestamping is performed in hardware or software. If timestamping is performed in software, it is subject to OS latencies that can vary significantly due to interrupt handling and scheduling~\cite{cochran:2011,DBLP:conf/cloud/LiSPG14}. Therefore, software timestamping results in a high discrepancy between the minimum and the maximum latency, thus drastically increasing the reading error. In order to mitigate a high reading error due to software timestamping, modern NICs come with timestamping support that implements egress and ingress timestamping of packets in hardware, thus minimizing timestamping jitter. We introduce functions $\mathcal{L}_{min}$ and $\mathcal{L}_{max}$ that return the respective minimum and maximum latency induced by a component on the communication path between $k$ and $k'$. As a result, we can identify the minimum and maximum transmit timestamping latency of node $k$ with $\mathcal{L}^{k}_{min}(txts)$ and $\mathcal{L}^{k}_{max}(txts)$ respectively.

After timestamping, the message is put on the medium and transmitted to the receiver~$k'$. The medium access latency depends on the link mode (full-duplex/half-duplex) and the arbitration mechanism in the case of half-duplex (e.g., carrier-sense multiple access). We denote the minimum and maximum medium access~($ma$) latency of the sending node $k$ with $\mathcal{L}^{k,k'}_{min}(ma)$ and $\mathcal{L}^{k,k'}_{max}(ma)$ respectively. The transmission of the message on the medium adds a minimum and maximum latency $\mathcal{L}^{k,k'}_{min}(c)$ and $\mathcal{L}^{k,k'}_{max}(c)$ that depend on the length of the cable and its material. Finally, node $k'$ receives the message and creates a reception timestamp. Again, just as with transmission timestamp creation, reception timestamp creation latency can vary depending on whether it is implemented in hardware or software. In accordance with our notation for transmit timestamping latency of node $k$ we write $\mathcal{L}^{k'}_{min}(rxts)$ and $\mathcal{L}^{k'}_{max}(rxts)$ for the minimum and maximum receive timestamping latency of node $k'$. Therefore, the minimum and maximum message delays between two directly connected physical nodes $k$ and $k'$ are given by the sums of minimum and maximum latency:
\begin{align}
d_{min}^{k,k'}&=\mathcal{L}^{k}_{min}(txts)+\mathcal{L}^{k,k'}_{min}(ma)+\mathcal{L}^{k,k'}_{min}(c)+\mathcal{L}^{k'}_{min}(rxts) \nonumber \\ d_{max}^{k,k'}&=\mathcal{L}^{k}_{max}(txts)+\mathcal{L}^{k,k'}_{max}(ma)+\mathcal{L}^{k,k'}_{max}(c)+\mathcal{L}^{k'}_{max}(rxts)
\end{align}

The introduction of virtualization into distributed real-time systems inserts a new abstraction layer and source of latency in the form of the hypervisor to the communication path. Depending on the implementation of clock synchronization and timestamping in the virtualized system, these latencies can add to the minimum and maximum message delays, which degrades the reading error. For the sake of providing a comprehensive discussion of how virtualization can affect the reading error, we, for now, presume that timestamping is performed in software inside a VM, resulting in a worst-case setup in which hypervisor latencies add to the message delay. As a result, we have to extend the formulas for the minimum $d_{min}^{k,k'}$ and maximum $d_{max}^{k,k'}$ message delay by an additional $\mathcal{L}^{n}_{min}(hv)$ and $\mathcal{L}^{n}_{max}(hv)$ with $n\in\{k,k'\}$ that captures minimum and maximum hypervisor latencies of~nodes~$k$~and~$k'$.

\setlength\extrarowheight{3pt}
\begin{table}[t]
	\centering
	\caption{Summary of the latencies that contribute to the minimum and maximum message delays on a communication path between two neighboring nodes $k$ and $k'$. The hypervisor latency $\mathcal{L}(hv)$ is divided by $\mathcal{L}(vme)$, $\mathcal{L}(sched)$, and $\mathcal{L}(vn)$.}
	\begin{tabularx}{\textwidth}{c|c|c}
		\shortstack{\textbf{\textit{Latency}}\\$x\in\{min, max\}$\\$n\in\{k,k'\}$} & \textbf{\textit{Description}} & \shortstack{\textbf{\textit{Magnitude}}\\\textbf{\textit{[sec]}}} \\
		\hline\hline
		$\mathcal{L}^{k}_{x}(txts)$ & \multicolumn{1}{|m{0.6\columnwidth}|}{\centering Min/max latency introduced by transmit timestamp creation in the sender $k$} & $10^{-7}$ to $10^{-5}$ \\
		$\mathcal{L}^{k,k'}_{x}(ma)$ & \multicolumn{1}{|m{0.6\columnwidth}|}{\centering Min/max latency introduced by the medium access} & $10^{-9}$ to $10^{-7}$ \\
		$\mathcal{L}^{k,k'}_{x}(c)$ & \multicolumn{1}{|m{0.6\columnwidth}|}{\centering Min/max latency introduced by transmission on the medium connecting nodes $k$ and $k'$} & $10^{-7}$ to $10^{-2}$ \\
		$\mathcal{L}^{k'}_{x}(rxts)$ & \multicolumn{1}{|m{0.6\columnwidth}|}{\centering Min/max latency introduced by reception \newline timestamp creation in the receiver $k'$} & $10^{-7}$ to $10^{-5}$ \\
		\hline\hline
		$\mathcal{L}^{n}_{x}(vme)$ & \multicolumn{1}{|m{0.6\columnwidth}|}{\centering Min/max latency introduced by VM exits} & $0$ to $10^{-3}$ \\
		$\mathcal{L}^{n}_{x}(sched)$ & \multicolumn{1}{|m{0.6\columnwidth}|}{\centering Min/max latency introduced by vCPU scheduling} & $0$ to $10^{-3}$ \\
		$\mathcal{L}^{n}_{x}(vn)$ & \multicolumn{1}{|m{0.6\columnwidth}|}{\centering Min/max latency introduced by virtual networking} & $10^{-5}$ to $10^{-2}$ \\
	\end{tabularx}
	\label{tab:reading_delay_resource_mappings}
\end{table}

To allow a more in-depth analysis of how virtualization affects the reading error, we further differentiate different types of hypervisor abstractions or resource mappings whose minimum and maximum latencies sum up to yield the minimum and maximum hypervisor latencies. We identify three different types of source of hypervisor latency:
\begin{description}[style=unboxed,leftmargin=0cm]
	\item[\textbf{VM Exits:}] If a VM accesses a virtual resource or when a physical resource triggers an interrupt for a specific VM, this causes a VM exit. A VM exit is a secure switch from the VM's unprivileged execution context to the privileged execution context of the hypervisor. The hypervisor determines the cause of the VM exit and handles the operation on behalf of the VM. For example, it mediates the access of the VM to an actual physical resource. What virtual resources are handled by the hypervisor directly is implementation-dependent. However, often they include virtual advanced interrupt controllers (vAPIC), Second Level Address Translation (SLAT), e.g., Intel's Extended Page Tables (EPT)~\cite{intel2020}, or virtual clocks and timers. The latency introduced by VM Exits can vary depending on the specific virtual resource triggering a VM exit and on the implementation of VM exit handling the hypervisor. In our experience, the processing time of a VM exit ranges from hundreds of nanoseconds to several milliseconds. We denote the minimum and maximum latency with $\mathcal{L}^{n}_{min}(vme)$ and $\mathcal{L}^{n}_{max}(vme)$ respectively. Notice that the absence of the same gives the minimum latency of a VM exit, so it is always $\mathcal{L}^{n}_{min}(vme)=0$ 
	\item[\textbf{vCPU Scheduling:}] The hypervisor maps the execution of vCPUs to pCPUs according to a scheduling policy. The time a vCPU is preempted and waiting to be dispatched by the scheduler yields the latency introduced by scheduling decisions. This time is also referred to as the worst-case response time of the used scheduling algorithm. The worst-case response time of a scheduling algorithm depends on the supposed task model and the task set. There is a rich literature on response time analysis that we refer the reader to~\cite{DBLP:conf/rtcsa/BalbastreRC09,DBLP:conf/emsoft/AlmeidaP04,DBLP:conf/dac/BeckertNEP14}. The magnitude of vCPU scheduling latencies depends on the scheduler's granularity and can range from a tenth of microseconds to milliseconds. We write $\mathcal{L}^{n}_{min}(sched)$ and $\mathcal{L}^{n}_{max}(sched)$ for the minimum and maximum vCPU scheduling latency. Again, it is always $\mathcal{L}^{n}_{min}(sched)=0$ since in the best case, there is no vCPU preemption at all. Note the difference between a VM exit and a preemption. In case of a VM exit, the hypervisor performs a privileged operation on behalf of the VM and accounts its execution time to the VM. In contrast, in the preemption case, the vCPU of the VM stops execution altogether and does not resume until the scheduler dispatches it again.
	\item[\textbf{Complex Virtual Devices:}] Other virtual resources, such as virtual networking or virtual disks, require the assistance of a privileged VM in order to provide more complex device emulation or to multiplex the access to a shared physical device using an interface that provides a simplified virtual device to co-located user VMs, such as the standardized VirtIO interface\footnote{\url{http://docs.oasis-open.org/virtio/virtio/v1.0/virtio-v1.0.html}}. The latency introduced by these complex assisted virtual resources is critical for clock synchronization if a virtual node utilizes a virtual NIC for connectivity. Virtual networking commonly adapts a frontend-backend architecture in which the frontend driver resides in the user VM, whereas the backend driver resides in the privileged VM and communicates to the actual device driver. The communication between frontend and backend is implemented as a shared memory with a notification mechanism, e.g., based on the injection of virtual interrupts into a VM. Therefore, the latency introduced by virtual networking depends on the processing time of network packets that are incoming or outgoing over the virtual network backend in the privileged VM and the latencies introduced by the notification mechanism. The processing time of a packet in the backend of the privileged VM depends strongly on the vCPU scheduling policy of the hypervisor.
	Furthermore, we have to account for OS latency since the send and receive timestamps of the frontend driver in the OS of a user VM are software-based. As a result of this complexity, virtual networking introduces high latency. Complexity by itself does not yet degrade clock synchronization precision. However, the interaction of various components residing in different VMs and privilege levels tend to result in a high discrepancy between minimum $\mathcal{L}^{n}_{min}(vn)$ and maximum $\mathcal{L}^{n}_{max}(vn)$ latency ranging from a tenth of microseconds up to several milliseconds~\cite{DBLP:conf/globecom/OljiraBTG16}.
\end{description}
We summarize all considered latencies and their approximate magnitudes in Table~\ref{tab:reading_delay_resource_mappings}. After discussing how the hypervisor latency is composed, we now provide formulas to obtain the minimum and the maximum hypervisor latency:
\begin{align}
\mathcal{L}_{min}^{n}(hv)&=\mathcal{L}^{n}_{min}(vn) \nonumber\\
\mathcal{L}_{max}^{n}(hv)&=\mathcal{L}^{n}_{max}(vme)+\mathcal{L}^{n}_{max}(sched)+\mathcal{L}^{n}_{max}(vn)
\end{align}
This let's us extend the formula for the minimum and maximum message delay by the minimum and maximum hypervisor latency introduced to nodes $k$ and $k'$ if they are virtualized:
\begin{align}
d_{min}^{k,k'}&=\mathcal{L}_{min}^{k}(hv)+\mathcal{L}^{k}_{min}(txts)+\mathcal{L}^{k,k'}_{min}(ma)+\mathcal{L}^{k,k'}_{min}(c)+\mathcal{L}^{k'}_{min}(rxts)+\mathcal{L}_{min}^{k'}(hv) \nonumber \\ d_{max}^{k,k'}&=\mathcal{L}_{max}^{k}(hv)+\mathcal{L}^{k}_{max}(txts)+\mathcal{L}^{k,k'}_{max}(ma)+\mathcal{L}^{k,k'}_{max}(c)+\mathcal{L}^{k'}_{max}(rxts)+\mathcal{L}_{max}^{k'}(hv)
\end{align}
If node $k$ or $k'$ are native, non-virtualized nodes, the terms for the hypervisor latencies simply equal zero.

So far, we have been considering the minimum and maximum message delay between two neighboring nodes. However, these formulas are easily extended to capture the minimum and maximum \emph{path message delay} between any two nodes $n,n'\in D$ of a virtualized distributed system~$D$. Let $E_{n,n'}$ be the path from a node $n\in D$ to a node $n'\in D$, whereas the link between two adjacent nodes $k,k'\in D$ of the path is represented by a tuple $(k,k')\in E_{n,n'}$.
\begin{definition}
	\label{def:min_max_message_delays_path}
	The minimum and maximum message delays $d_{min}^{n,n'}$ and $d_{max}^{n,n'}$ on path $E_{n,n'}$ are then given by the sum of the minimum and maximum message delays $d_{min}^{k,k'}$ and $d_{max}^{k,k'}$ of all adjacent nodes $(k,k')\in E_{n,n'}$ on said path:
	\begin{align}
	d_{min}^{n,n'}=\sum_{(k,k')\in E_{n,n'}}d_{min}^{k,k'}\hspace{1cm}\text{ and }\hspace{1cm}d_{max}^{n,n'}=\sum_{(k,k')\in E_{n,n'}}d_{max}^{k,k'}
	\end{align}
\end{definition}
As a result of Definition~\ref{def:min_max_message_delays_path}, we can determine the minimum, and maximum message delays $d_{min}$ and $d_{max}$ for the entire virtualized distributed system~$D$ according to Definition~\ref{def:reading_error} by finding the minimum and maximum message delays of all possible paths in $D$.
\begin{definition}
	\label{def:min_max_message_delay}
	Let $E_D$ be the set of all possible paths between two nodes in the virtualized distributed system $D$. The minimum and maximum message delays $d_{min}$ and $d_{max}$ that yield the reading error of $D$ according to Definition~\ref{def:reading_error} are then given by:
	\begin{align}
	d_{min}=\min_{\forall E_{n,n'}\in E_D}(d_{min}^{n,n'})\hspace{1cm}\text{ and }\hspace{1cm}d_{max}=\max_{\forall E_{n,n'}\in E_D}(d_{max}^{n,n'})
	\end{align}
\end{definition}
Definition~\ref{def:min_max_message_delay} provides us with several insights on how the precision of clock synchronization in virtualized distributed systems can degrade. First of all, regardless of virtualized or native distributed systems, we can see that big networks with varying path lengths cause a high discrepancy between the induced minimum and maximum latency of a message transmission leading to an increased reading error. Heterogeneous nodes and network technologies intensify the negative effects of asymmetric networks since the minimum and maximum latencies of timestamp creation or medium access can vary significantly from hardware platform to hardware platform just as the transmission latency varies with the used medium, e.g., a glass fiber compared to a twisted-pair copper cable. Virtualization adds to these latencies. In particular, virtual networking can negatively affect the reading error due to its complexity and resulting in highly variable latencies.

Conversely, careful network planning and design and use of preferably homogeneous hardware platforms for the network nodes can mitigate a high reading error. When it comes to moderating the effects of virtualization on the reading error, one possibility is to reduce the discrepancy between the minimum and maximum latencies induced by the hypervisor, e.g., by accounting for latency introduced by virtual networking during runtime~\cite{yao2017}. However, a straightforward approach for ultimately mitigating hypervisor latency is to utilize hardware timestamping by passing through a NIC to a dedicated VM. Following the dependent clock paradigm, this VM can then share its synchronized time with co-located VMs and the hypervisor, thus enabling coordinated time-triggered scheduling of VCPUs according to a global VCPU, task, and network schedule.

\subsection{A Dependent Clock for the ACRN Hypervisor}
\label{ssec:dependent_clock_paradigm}

The dependent clock paradigm mitigates the drawbacks of executing one instance of a clock synchronization protocol per VM~\cite{DBLP:conf/osdi/BroomheadCRV10} by performing clock synchronization only in a single VM and sharing the synchronized time with co-located VMs on the same host. This is beneficial with respect to hypervisor latencies that degrade the reading error, since a dedicated \emph{clock synchronization VM} can be granted direct access to a physical NIC effectively preventing any latencies due to virtual networking, VM exits, and vCPU scheduling since clock synchronization can resort to hardware timestamping yielding $\mathcal{L}_{min}(vn)=\mathcal{L}_{max}(vn)=0$, $\mathcal{L}_{min}(vme)=\mathcal{L}_{max}(vme)=0$, and $\mathcal{L}_{min}(sched)=\mathcal{L}_{max}(sched)=0$ (cf.~Section~\ref{ssec:degradation_of_clock_sync_precision}).

Typically, sharing the synchronized time with co-located VMs is achieved through a shared memory region. This makes the hypervisor access to the synchronized time as easy as mapping the shared memory region into the hypervisors address space. The hypervisor can then utilize the synchronized time and an upper bound on the precision of clock synchronization to derive a global time~\cite{DBLP:books/sp/Kopetz11} whose marcoticks drive the vCPU scheduling of a time-triggered hypervisor. Our dependent clock implementation consists of three sub-components -- a shared memory implementation in the hypervisor, a kernel module in the Linux kernel of the guest OS, and a user space component performing the actual clock synchronization.

The ACRN hypervisor~\cite{DBLP:conf/vee/LiXRD19} is a type-1 hypervisor that is configured statically at compile time. ACRN's configuration includes vCPU configuration with pCPU affinity, main memory partitioning, and device configuration that is either device passthrough or device virtualization. ACRN's device virtualization comprises a shared memory device \lstinline|ivshmem| that allows the configuration of shared memory regions between VMs at compile time\footnote{\url{https://projectacrn.github.io/2.0/developer-guides/hld/ivshmem-hld.html}}. The shared memory device is implemented as a virtual PCI device that either resides in the service OS or the hypervisor and exposes the base address and size of the shared memory region to a VM.

Our goal was to utilize a shared memory region between VMs and the hypervisor to expose clock synchronization parameters that yield the synchronized time in combination with the invariant TSC. After studying the implementation of the \lstinline|ivshmem| shared memory device to assess its suitability for our purposes, we decided not to use the existing implementation but instead implement a separate shared memory device exposing the clock synchronization parameters for several reasons. We were mainly concerned with access control and the semantics of this particular shared memory. The existing implementation was intended for bidirectional inter-VM communication. We aimed to expose the clock synchronization parameters of a dedicated clock synchronization VM to co-located VMs, including the hypervisor. Only this particular VM should have write access to the shared memory. Even though possible, utilizing \lstinline|ivshmem| would have required us to make the general shared memory implementation aware of the specifics of our use case, thus breaking the modularity emphasized by the ACRN project\footnote{\url{https://projectacrn.github.io/2.0/developer-guides/modularity.html}}.

As a result, we implemented a new synchronized time shared memory device \lstinline|stshmem| located in the hypervisor. The virtual device reassembles the \lstinline|ivshmem| device in that it is a virtual PCI that exposes a shared memory's base address and size to preconfigured VMs. In contrast to the \lstinline|stshmem| device, though, its compile-time configuration grants only a dedicated clock synchronization VM write access in order to be able to update the clock synchronization parameters. The shared clock synchronization parameters include a time multiplier, a shift value, the synchronized time at the last update, and the TSC value at the last update of the parameters. A Linux kernel module discovers the PCI device and retrieves and exports the shared memory base address and size to the Linux timekeeping subsystem. Linux timekeeping utilizes the clock parameters and the current TSC value to implement a para-virtualized clock that yields the synchronized time. Following existing Linux kernel clocks such as \lstinline|CLOCK_MONOTONIC| or \lstinline|CLOCK_REALTIME|, the para-virtualized clock is made available to kernel and user space as \lstinline|CLOCK_SYNCTIME|.

In the user space of the clock synchronization VM, we utilize Linux PTP to execute the gPTP clock synchronization protocol. The clock synchronization VM has dedicated access to the NIC that connects internally to the TSN switch. The \lstinline|ptp4l| clock synchronization application processes incoming clock synchronization messages and synchronizes the clock of the NIC. The \lstinline|phc2sys| application in turn interfaces with the Linux kernel's NTP subsystem that allows for the synchronization of Linux \lstinline|CLOCK_REALTIME| to the NIC's clock. We extended the NTP subsystem to update the clock parameters of \lstinline|CLOCK_REALTIME| and the clock synchronization parameters of \lstinline|CLOCK_SYNCTIME| in our shared memory region. As a result, co-located VMs can transparently retrieve the synchronized time by reading \lstinline|CLOCK_SYNCTIME|. The hypervisor can implement a synchronized clock equivalently to the Linux kernel by utilizing the clock synchronization parameters in the shared memory region.

\section{Analysis of IEEE\,802.1AS in a Virtualized Distributed System}
\label{sec:analysis_ieee802.1as_virtualized_distributed_system}
We will now use our formalization and our understanding of how virtualization affects the reading error to analyze IEEE\,802.1AS~\cite{ieee8021as} and derive an upper bound on the clock synchronization precision of IEEE~802.1AS in a virtualized distributed real-time system. At the same time, this shows the applicability of the formal model to real-world protocols.

\subsection{IEEE\,802.1AS}
IEEE~802.1AS, henceforth also referred to as gPTP, is a profile of PTP, which is specified in IEEE~1588-2008~\cite{ieee1588-2008}. It is a hierarchical clock synchronization protocol. The grandmaster clock $g\in D$ is elected by all so-called time-aware nodes using the best master clock algorithm (BMCA) specified in the standard~\cite{ieee1588-2008,ieee8021as}. The nodes of a distributed system synchronize their clocks to the grandmaster by periodically exchanging a series of messages. As before, we denote the resynchronization period with $S$ and a specific resynchronization period with $s$.

Furthermore, we differentiate two classes of gPTP messages. The first class of messages is only exchanged between neighboring nodes $n$ and $n'$ to estimate the propagation delay of a specific link. For a detailed description of the propagation delay measurement, we refer the reader to the standard~\cite{ieee8021as}.

The second class of messages is used to synchronize the nodes to the grandmaster. Synchronization information originates from the grandmaster and traverses the synchronization hierarchy. We consider gPTP in two-step mode. The grandmaster $g\in D$ periodically sends out a $Sync$ and a $Follow\_Up$ message to its neighbors $n\in N^g$ and stores the transmission timestamp $tn^g(tx_{Sync})$ of the $Sync$ message. The messages are received by the neighboring nodes, being processed, and forwarded down the synchronization hierarchy until they reach a leaf node. We examine a node $d$ in the clock synchronization hierarchy. Let $E_d$ be the path from the grandmaster node $g\in D$ to node $d$. Upon reception of a $Sync$ message from its predecessor node $d'\in D$ on the path $E_d$, node $d$ creates a reception timestamp $tn^{d}(rx_{Sync})$. The $Sync$ message itself is empty. Using the semantics introduced in Definition~\ref{def:sync_message} we can write it as $Sync:=(0,0,\emptyset)$. Next, node $d$ receives the $Follow\_Up$ message from node $d'$:
\begin{align}
Follow\_Up:=(e^{g,d'},0,\{tn^g(tx_{Sync}),R_{\frac{g}{d'}}\})
\end{align}
The $Follow\_Up$ message contains the actual synchronization information including the \textit{precise origin timestamp} that is the transmission timestamp $tn^g(tx_{Sync})$ of the $Sync$ message at the grandmaster $g$, the \textit{correction field} $e^{g,d'}$, and the $rateRatio$~$R_{\frac{g}{d'}}$. The correction field $e^{g,d'}$ is the accumulated propagation delay $e^{n,n'}$ of all neighboring nodes $(n,n')\in E_{d'}$ plus the residence time $T_n$ in each node $n$ along the path $E_{d'}$:
\begin{align}
e^{g,d'}=\sum_{\forall (n,n')\in E_{d'}}(e^{n,n'}+T^{n'})
\end{align}
We can derive the residence time in each node $n\in D$ from the reception and transmission timestamps of the $Sync$ message of the respective node and transform it into the time scale of the grandmaster clock:
\begin{align}
T^{n'}=R_{\frac{g}{n'}}(tn^{n'}(tx_{Sync})-tn^{n'}(rx_{Sync}))
\end{align}
Finally, node $d$ can use the synchronization information of the $Follow\_Up$ message, its reception timestamp $tn^{d}(rx_{Sync})$, and the propagation delay to the grandmaster $e^{g,d}=e^{g,d'}+e^{d',d}$ to derive a correction term $c_s^{g,d}$ for the synchronization period $s$ that is the approximated offset of the local clock $C^{d}$ to the grandmaster clock~$C^g$:
\begin{align}
\label{eq:1as_correction_term}
c_s^{g,d}=R_{\frac{g}{d}}\cdot tn^{d}(rx_{sync})-(tn^g(tx_{Sync})+e^{g,d}_s)
\end{align}

\subsection{The Reading Error of IEEE\,802.1AS}
\label{ssec:reading_error_ieee802.1as}
We have learned in Definition~\ref{def:correction_term} that the correction term given in Equation~(\ref{eq:1as_correction_term}) is afflicted with a reading delay $\epsilon_s^{g,d}$ so that it yields the actual time difference $\Delta^{g,d}_s$ between the grandmaster $g$ and the node $d$ plus the reading delay $c_s^{g,d}=\Delta^{g,d}_s+\epsilon_s^{g,d}$. Furthermore, given Inequality~(\ref{eq:min_max_message_delay}) and Definition~\ref{def:reading_error}, for any two reading delays $\epsilon_s^{g,d'}$ and $\epsilon_s^{g,d^*}$ we known that it is $|\epsilon_s^{g,d'} - \epsilon_s^{g,d^*}| \le \mathcal{E}$.

As a result, if we can derive a formula for the reading error of IEEE\,802.1AS, we found an upper bound of clock synchronization precision according to the clock synchronization condition (c.f.~Theorem~\ref{theorem:clock_synchronization_condition}). A minor tweak to Definition~\ref{def:min_max_message_delay} allows us to derive the reading error of IEEE\,802.1AS in a virtualized distributed real-time system. We have established in Section~\ref{ssec:degradation_of_clock_sync_precision} that the minimum and maximum message delays $d_{min}$ and $d_{max}$, which yield the reading error, are given by the minimum and maximum message delays of all possible paths~$E_D$ in a distributed system~$D$. In case of gPTP, we are not interested in bounding arbitrary reading delays $\epsilon^{n,n'}$ and $\epsilon^{n,n^*}$ (c.f.~Inequality~\ref{eq:min_max_message_delay}) but the reading delays $\epsilon^{g,n'}$ and $\epsilon^{g,n^*}$ with the grandmaster $g$. Therefore, we determine the minimum and maximum message delay of all possible paths $E_G\subset E_D$ from the grandmaster $g$ to its subordinate nodes $d\in D\setminus\{g\}$.
\begin{definition}[Reading Error of IEEE\,802.1AS]
	\label{def:reading_error_ieee802.1as}
	Let $E_G$ be the set of all paths between the grandmaster $g$ and its subordinate nodes $d\in D\setminus\{g\}$. The minimum and maximum message delays $d_{min}$ and $d_{max}$ are given by:
	\begin{align}
	d_{min}=\min_{\forall E_{n,n'}\in E_G}(d_{min}^{n,n'})\hspace{1cm}\text{ and }\hspace{1cm}d_{max}=\max_{\forall E_{n,n'}\in E_G}(d_{max}^{n,n'})\nonumber
	\end{align}
	The reading error $\mathcal{E}$ of IEEE\,802.1AS in a virtualized distributed real-time system $D$ is then given by:
	\begin{align}
	\mathcal{E}=d_{max}-d_{min}\nonumber
	\end{align}
\end{definition}
Knowing the reading error, we can derive an upper bound of clock synchronization precision for IEEE\,802.1AS in a virtualized distributed real-time system using the clock synchronization condition~(c.f.~Theorem~\ref{theorem:clock_synchronization_condition}).

\section{Experimental Evaluation}
\label{sec:evaluation}
We now present an experimental evaluation of our formal model and show its applicability to real-world systems. To that end, we first introduce our experimental setup and illustrate how our formal model can be applied to derive an upper bound on clock synchronization precision. Next, we describe how we measure the clock synchronization precision and discuss how this measurement is afflicted with an inherent measurement error. Finally, we present and discuss the experimental results, including a baseline experiment featuring native systems without virtualization, a resource consolidating hypervisor configuration, and a resource partitioning hypervisor configuration.

\subsection{Experimental Setup}
\label{ssec:experimental_setup}
We performed our experiments on three commercial edge computing devices\footnote{\url{https://www.tttech-industrial.com/wp-content/uploads/TTTech-Industrial_MFN-100.pdf}} ($mfn_1$, $mfn_2$, $mfn_3$). The devices come with an Intel~Atom~E3950 multicore processor with four cores at 1.594\,GHz each, 8\,GB main memory, and two Intel I210 NICs $NIC_1$ and $NIC_2$. $NIC_1$ connects internally to an integrated Linux-based TSN switch with four external Gigabit ports. We refer to the switches in $mfn_1$, $mfn_2$, and $mfn_3$ with $sw_1$, $sw_2$, and $sw_3$ respectively. The ethernet controller used in the TSN switches is an Intel I210 as well. The switch $sw_3$ of $mfn_3$ is configured to act as grandmaster clock ($gm$). All switches run a full preemption Linux kernel v4.9.78\footnote{\url{https://wiki.linuxfoundation.org/realtime/start}}. Both host nodes $mfn_1$ and $mfn_2$ connect to $sw_3$, the grandmaster clock, via their respective switches $sw_1$ and $sw_2$. We want to point out the symmetry and homogeneity of our experimental setup resulting in the exact same paths lengths and, except for minor variations, identical timing properties of intermediate nodes from the grandmaster $gm$ to the leaf nodes $mfn_1$ and $mfn_2$. All nodes of our experimental distributed real-time system execute OpenIL's\footnote{\url{https://www.openil.org/}} fork of LinuxPTP\footnote{\url{https://github.com/openil/linuxptp}} (v1.8) performing clock synchronization according to IEEE\,802.1AS. In case of the native setup, nodes $mfn_1$ and $mfn_2$ run Ubuntu Server~20.04.1 with a full preemption Linux kernel~v4.19.72. The network topology and configuration of the native setup is illustrated in Figure~\ref{fig:experiment_configuration_baseline}.
\begin{figure}[t]
	\centering
	\includegraphics[width=5in,trim={0.2cm 0.05cm 0.15cm 0.05cm},clip]{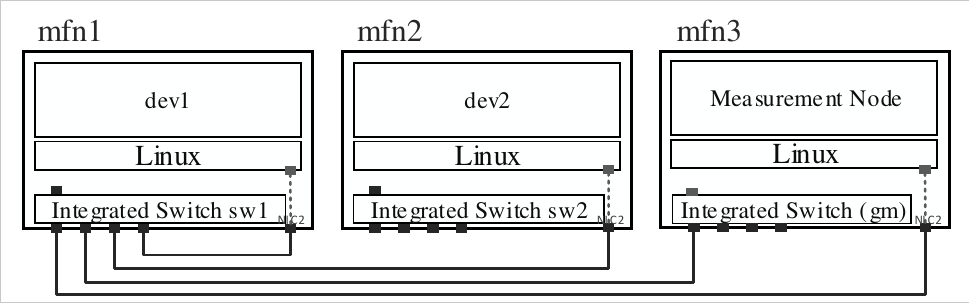}
	\caption{Network topology and node configuration of our baseline experiments. The edge computing devices $mfn_1$ and $mfn_2$ come with an integrated TSN switch that connects them to the switch of $mfn_3$ that acts as grandmaster clock ($gm$). Note that in the native setup, $mfn_2$ connects to the grandmaster $gm$ via the switch of $mfn_1$.}
	\label{fig:experiment_configuration_baseline}
\end{figure}

\begin{figure}[t]
	\centering
	\includegraphics[width=5in,trim={0.1cm 0.05cm 0.4cm 0.05cm},clip]{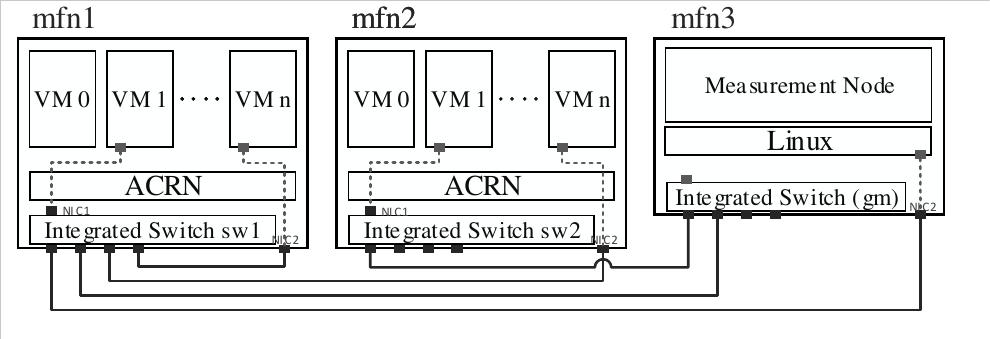}
	\caption{Network topology and node configuration of our virtualized experiments. VM $vm_0$ acts as privileged service VM, VM $vm_1$ as clock synchronization VM, and VM $vm_n$ as measurement VM on their respective node.}
	\label{fig:experiment_configuration}
\end{figure}

For the virtualized setup, nodes $mfn_1$ and $mfn_2$ run the ACRN hypervisor~v2.2 with our modifications to support a dependent clock. We indicate that a VM $vm_x$ runs on host $mfn_y$ by writing $mfn_y^{vm_x}$. On both virtualized nodes $mfn_1$ and $mfn_2$, $vm_0$ acts as ACRN's privileged service VM. It is configured with 2\,GiB of main memory, a single vCPU, and runs Ubuntu Server~20.04.1 with Linux kernel v5.4.59. For all other VMs, we use a configuration with 1\,GiB of main memory, and a single vCPU. They run Ubuntu Server~18.04.3 with a full preemption Linux kernel~v4.19.72. We pass-through $NIC_2$ of the respective edge computing device to $vm_n$ and $NIC_1$, which connects internally to the integrated switch, to $vm_1$. Figure~\ref{fig:experiment_configuration} illustrates the network topology and configuration of our virtualized setup. Furthermore, all VMs have access to the raw timestamp counter (TSC) of the processor, which thus provides a discontinuous virtual clock with $\nu(l)=l-l_0$ (cf.~Section~\ref{ssec:virtual_clocks}). The TSC-based discontinuous virtual clock fulfills the virtual clock condition and enables the operation of our dependent clock that provides a synchronized time to all VMs.

\subsection{Model Instantiation}
\label{ssec:latency_measurement}
In Section~\ref{ssec:degradation_of_clock_sync_precision}, we have discussed how clock synchronization precision in virtualized distributed real-time systems degrades with the reading error. As a result, to derive an upper bound on the precision according to the clock synchronization condition (c.f.~Theorem~\ref{theorem:clock_synchronization_condition}), we have to instantiate the variables of our model for all experiments. In the following, we will describe how to apply our model to a real system, derive the drift offset, and calculate the reading error.

\subsubsection{Drift Offset}
We can obtain the drift offset by substituting the resynchronization period $S$ and the maximum drift rate $r_{max}$ observed in our experimental setup. The resynchronization period $S$ is given by the configuration of the \lstinline|logSyncInterval| parameter of IEEE\,802.1AS which in our case equals to $S=2^{-3}\,s=125\,ms$. In the case of the maximum drift rate, we had difficulties obtaining an accurate value due to the lack of information on the oscillators used for the Intel I210 NICs and the Intel Atom processors in our edge computing devices. Therefore, we could either fallback to the maximum tolerated value for the drift rate specified in the IEEE\,802.1AS standard~\cite{ieee8021as} with $r_{max}=5\cdot 10^{-6}\,ppm$, which seemed too pessimistic, or measure the drift rate between $mfn_1$ and $mfn_2$ in a non-virtualized clock synchronization setup.

\newcolumntype{Y}{>{\centering\arraybackslash}X}
\begin{table}[t!]
	\setlength\tabcolsep{3pt}
	\centering
	\caption{Overview of minimum and maximum latencies that are used to obtain the path message delays and finally the reading error of a specific experimental setup.}
	\begin{tabularx}{\textwidth}{c|c|Y}
		Latencies & Nodes & Setups\\
		\hline\hline
		$\mathcal{L}^n_{min}(txts)=984\,ns$ & \multirow{2}{*}{\shortstack{$n\in\{mfn_1$, $mfn_2$, $sw_1$,\\ $sw_2$, $sw_3$/$gm\}$}} & \multirow{2}{*}{native HWTS} \\
		$\mathcal{L}^n_{max}(txts)=1024\,ns$ & & \\
		\hline
		$\mathcal{L}^n_{min}(rxts)=2148\,ns$ & \multirow{2}{*}{\shortstack{$n\in\{mfn_1$, $mfn_2$, $sw_1$, $sw_2$,\\ $sw_3$/$gm$, $cs_1$, $cs_2\}$}} & \multirow{2}{*}{\shortstack{native HWTS,\\ACRN BVT/VM exit-less HWTS}} \\
		$\mathcal{L}^n_{max}(rxts)=2228\,ns$ & & \\
		\hline\hline
		$\mathcal{L}_{min}^{n,n'}(ma)+\mathcal{L}_{min}^{n,n'}(c)=261\,ns$ & \multirow{2}{*}{$n,n'\in\{sw_3$/$gm$, $sw_1\}$} & \multirow{2}{*}{all} \\
		$\mathcal{L}_{max}^{n,n'}(ma)+\mathcal{L}_{max}^{n,n'}(c)=286\,ns$ & & \\
		\hline
		$\mathcal{L}_{min}^{n,n'}(ma)+\mathcal{L}_{min}^{n,n'}(c)=147\,ns$ & \multirow{2}{*}{$n,n'\in\{sw_1$, $mfn_2$/$mfn_2^{vm_n}\}$} & \multirow{2}{*}{all} \\
		$\mathcal{L}_{max}^{n,n'}(ma)+\mathcal{L}_{max}^{n,n'}(c)=179\,ns$ & & \\
		\hline
		$\mathcal{L}_{min}^{n,n'}(ma)+\mathcal{L}_{min}^{n,n'}(c)=142\,ns$ & \multirow{2}{*}{$n,n'\in\{sw_1$, $mfn_1$/$mfn_1^{vm_n}\}$} & \multirow{2}{*}{all} \\
		$\mathcal{L}_{max}^{n,n'}(ma)+\mathcal{L}_{max}^{n,n'}(c)=182\,ns$ & & \\
		\hline
		$\mathcal{L}_{min}^{n,n'}(ma)+\mathcal{L}_{min}^{n,n'}(c)=134\,ns$ & \multirow{2}{*}{$n,n'\in\{sw_1$, $mfn_3$/$mn\}$} & \multirow{2}{*}{all} \\
		$\mathcal{L}_{max}^{n,n'}(ma)+\mathcal{L}_{max}^{n,n'}(c)=175\,ns$ & & \\
		\hline
		$\mathcal{L}_{min}^{n,n'}(ma)+\mathcal{L}_{min}^{n,n'}(c)=462\,ns$ & \multirow{2}{*}{$n,n'\in\{sw_1$, $cs_1\}$} & \multirow{2}{*}{all} \\
		$\mathcal{L}_{max}^{n,n'}(ma)+\mathcal{L}_{max}^{n,n'}(c)=506\,ns$ & & \\
		\hline
		$\mathcal{L}_{min}^{n,n'}(ma)+\mathcal{L}_{min}^{n,n'}(c)=441\,ns$ & \multirow{2}{*}{$n,n'\in\{sw_2$, $cs_1\}$} & \multirow{2}{*}{all} \\
		$\mathcal{L}_{max}^{n,n'}(ma)+\mathcal{L}_{max}^{n,n'}(c)=502\,ns$ & & \\
	\end{tabularx}
	\label{table:latency_substitutions}
\end{table}

We decided to attempt the latter and measure the drift rate. For this measurement, we ran Ubuntu~20.04.1 with full preemption Linux kernel v4.19.72 on both nodes as depicted in Figure~\ref{fig:experiment_configuration_baseline}. After initial synchronization of the Intel I210 NIC's in both nodes using LinuxPTP, we let its clocks drift apart for one hour and measured the time offset, which yielded $85279\,\mu s$. This results in a rough estimate of the maximum drift rate between the Intel I210 NICs of $r_{max}\approxeq1.184*10^{-8}$. We are aware that this value is only an approximation and does not reflect the drift rate between the TSCs of the edge computing devices when it comes to the physical clock that drives our \lstinline|CLOCK_SYNCTIME| as discussed in Section~\ref{ssec:dependent_clock_paradigm}. However, we consider the drift between the two Intel NIC's clock to indicate the magnitude of the actual maximum drift rate. Furthermore, we argue that for our experiments, the resulting drift offset $\Gamma=2\cdot 125\,ms\cdot 1.184\cdot 10^{-8}=2.96\,ns$ of a few nanoseconds is neglectable compared to the degradation of clock synchronization due to the reading error.

\subsubsection{Reading Error}
The reading error $\mathcal{E}$ of IEEE\,802.1AS is given by the difference of the maximum $d_{max}$ and minimum $d_{min}$ message delay encountered in a distributed real-time system (c.f.~Definition~\ref{def:reading_error_ieee802.1as}). Let's start off with quantifying all latencies on the paths between the grandmaster $gm$ and nodes $mfn_1$ and $mfn_2$ in order to derive the reading error for the native setup as depicted in Figure~\ref{fig:experiment_configuration_baseline}. In this case, $E_G:=\{E_{gm,mfn_1},E_{gm,mfn_2}\}$ is the set of paths from the grandmaster $gm$ to nodes $mfn_1$ and $mfn_2$. The minimum $d_{min}$ and maximum $d_{max}$ message delays are then given by:
\begin{align}
d_{min}=\min(d_{min}^{gm,mfn_1},d_{min}^{gm,mfn_2})\hspace{1cm}\text{ and }\hspace{1cm}d_{max}=\max(d_{max}^{gm,mfn_1},d_{max}^{gm,mfn_2})
\end{align}

We pick $d_{max}^{gm,mfn_2}$ as an example of how to calculate the minimum or maximum path message delay using Definition~\ref{def:min_max_message_delays_path}. The path between the grandmaster $gm$ and node $mfn_2$ is given by $E_{gm,mfn_2}:=\{(gm,sw_1),(sw_1,mfn_2)\}$. As a result, for the maximum path message delay $d_{max}^{gm,mfn_2}$, we obtain:
\begin{align}
d_{max}^{gm,mfn_2}&=d_{max}^{gm,sw_1}+d_{max}^{sw_1,mfn_2}\nonumber\\
&=\mathcal{L}_{max}^{gm}(txts)+\mathcal{L}_{max}^{gm,sw_1}(ma)+\mathcal{L}_{max}^{gm,sw_1}(c)+\mathcal{L}_{max}^{sw_1}(rxts)\nonumber\\
&\hspace{0.3cm} +\mathcal{L}_{max}^{sw_1}(txts)+\mathcal{L}_{max}^{sw_1,mfn_2}(ma)+\mathcal{L}_{max}^{sw_1,mfn_2}(c)+\mathcal{L}_{max}^{mfn_2}(rxts)
\end{align}
The remaining path message delays $d_{min}^{gm,mfn_1},d_{min}^{gm,mfn_2},$ and $d_{max}^{gm,mfn_1}$ can be calculated accordingly.

In the virtualized setups, the respective clock synchronization VMs replace the native edge computing devices in the paths from the grandmaster $gm$ to the leaf nodes. For better readability, we denote the clock synchronization VMs with $cs_1:=mfn_1^{vm_1}$ and $cs_2:=mfn_2^{vm_1}$. The paths from the grandmaster to the leaf nodes are now given by $E_G:=\{E_{gm,cs_1},E_{gm,cs_2}\}$ with $E_{gm,cs_1}:=\{(gm,sw_1),(sw_1,cs_1)\}$ and $E_{gm,cs_2}:=\{(gm,sw_2),(sw_2,cs_2)\}$ as depicted in Figure~\ref{fig:experiment_configuration}. The minimum and maximum message delays and the reading error can be calculated as demonstrated for the native setup.	

\subsection{Clock Synchronization Precision Measurement}
\label{ssec:clock_sync_precision_measurement}
In general, we measure the precision of clock synchronization in a distributed system by timestamping the occurrence of an external event on all nodes of the distributed system and compare the timestamps. In our setup, a UDP multicast packet sent every second from the measurement node ($mn$) to all nodes of our distributed system acts as an external event. The worst-case offset of two node's reception timestamps of the same UDP multicast packet then yields the clock synchronization precision. However, this measurement suffers from an inherent measurement error $\gamma$ that originates, just as the reading error determines clock synchronization precision, from variable path message delays.

For the native baseline experiment with hardware timestamping, this measurement error boils down to a variable minimum and maximum latency of the used networking hardware analog to the minimum and maximum latency constituting to the reading error between the measurement node $mn$ and the respective leaf nodes $mfn_1$ and $mfn_2$. As a result, the expected measurement error of the native setup $\gamma$ is given by the difference of the maximum and minimum path message delay on the paths $E_{mn}$ between the measurement node $mn$ and the nodes $mfn_1$ and $mfn_2$:
\begin{equation}
\gamma=\max_{\forall E_{n,n'}\in E_{mn}}(d_{max}^{n,n'})-\min_{\forall E_{n,n'}\in E_{mn}}(d_{min}^{n,n'})
\end{equation}
Using the minimum and maximum latencies provided in Table~\ref{table:latency_substitutions} we calculate an inherent measurement error of $\gamma=321\,ns$ for the native setup. This measurement error can cause isolated data points of the measured clock synchronization precision to exceed the upper bound calculated using our formal model. We could counteract this by simply including the measurement error in the upper bound. However, this would be semantically incorrect since the measurement of clock synchronization precision and the distributed execution of the clock synchronization protocol are two independent activities. Therefore, we decided to plot the measurement error $\gamma$ instead and acknowledge the possibility of isolated data points that can exceed the upper bound.

\begin{figure}[t]
	\subfigure[Virtualization measurement error $\gamma_{virt}$ of $mfn_1^{vm_4}$]{
		\includegraphics[width=0.5\textwidth]{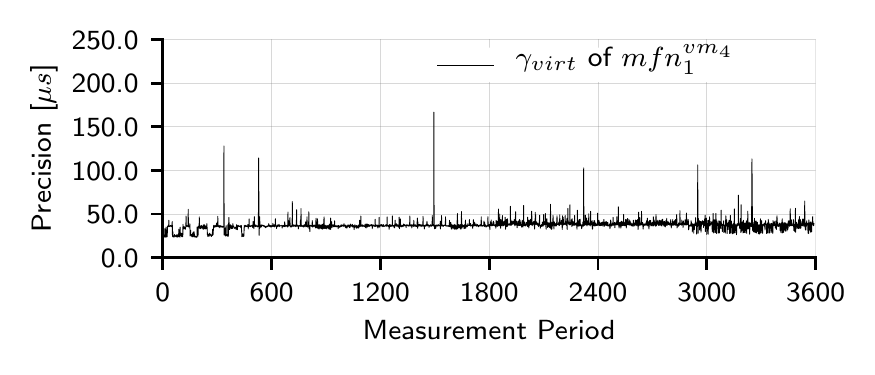}
		\label{fig:bvt_mfn1_measurement_error}
	}\hfill
	\subfigure[Virtualization measurement error $\gamma_{virt}$ of $mfn_2^{vm_4}$]{
		\includegraphics[width=0.5\textwidth]{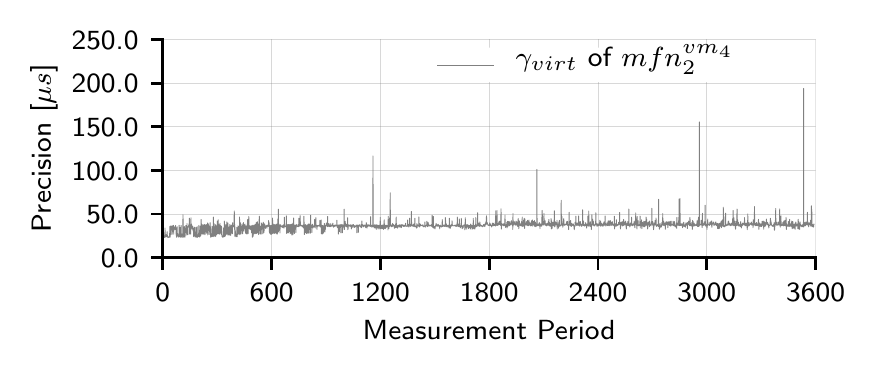}
		\label{fig:bvt_mfn2_measurement_error}
	}
	\subfigure[Virtualization measurement error $\gamma_{virt}$ of $mfn_1^{vm_3}$]{
		\includegraphics[width=0.5\textwidth]{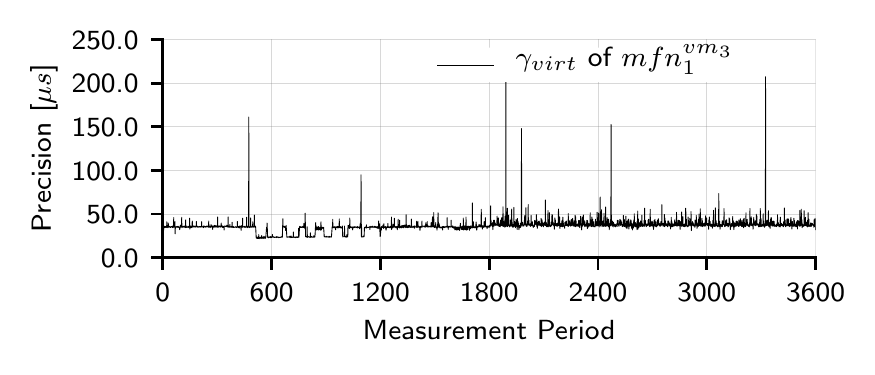}
		\label{fig:vmexitless_mfn1_measurement_error}
	}\hfill
	\subfigure[Virtualization measurement error $\gamma_{virt}$ of $mfn_2^{vm_3}$]{
		\includegraphics[width=0.5\textwidth]{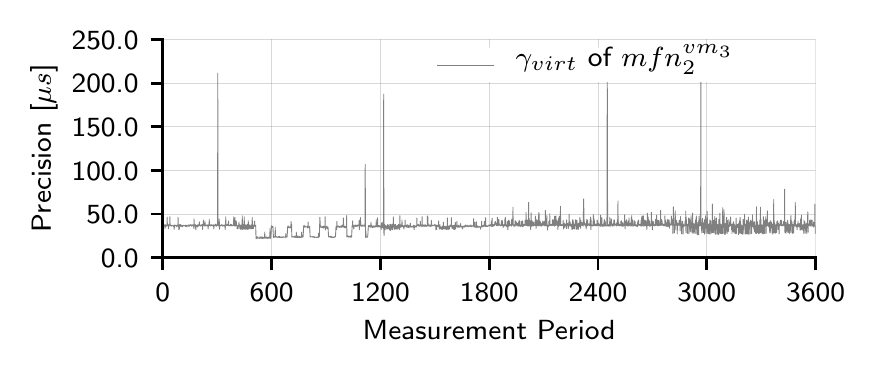}
		\label{fig:vmexitless_mfn2_measurement_error}
	}
	\caption{Showing the virtualization measurement error due to hypervisor and guest OS latency. Figures~\ref{fig:bvt_mfn1_measurement_error}~and~\ref{fig:bvt_mfn2_measurement_error} show the measured values of our resource consolidating hypervisor setup for measurement VMs $mfn_1^{vm_4}$ and $mfn_2^{vm_4}$. Figures~\ref{fig:bvt_mfn1_measurement_error}~and~\ref{fig:bvt_mfn2_measurement_error} in turn show the measured values of our resource partitioning hypervisor setup for measurement VMs $mfn_1^{vm_3}$ and $mfn_2^{vm_3}$. We can see that the measurement error is independent of the particular hypervisor configuration since the respective measurement VMs are always isolated on processor core $3$.}
	\label{fig:virt_measurement_error}
\end{figure}

For the virtualized experiments, we introduce a \emph{measurement VM} that receives and timestamps measurement packets using the node global synchronized time \lstinline|CLOCK_SYNCTIME| provided by the clock synchronization VM. It is beneficial to utilize a dedicated measurement VM for clock synchronization precision measurement. This way, we are guaranteed to measure the precision of the synchronized time provided by \lstinline|CLOCK_SYNCTIME| through the shared memory and not only the precision of the synchronized time available to the clock synchronization VM itself. The measurement VM $vm_n$ on the respective nodes $mfn_1$ and $mfn_2$ has exclusive access to $NIC_2$ (cf.~Figure~\ref{fig:experiment_configuration}) and is always pinned to processor core $3$. The measurement VM's kernel is modified to create a software timestamp from \lstinline|CLOCK_SYNCTIME| instead of using the hardware timestamp of the NIC for incoming packets. As a result, in addition to the inherent measurement error $\gamma$, due to path message delays, the clock synchronization precision measurement is subject to a measurement error $\gamma_{virt}$ stemming from variable guest OS and hypervisor latencies (c.f.~Section~\ref{ssec:degradation_of_clock_sync_precision}). However, in contrast to the inherent measurement error $\gamma$, we can exactly determine and account for the measurement error $\gamma_{virt}$ of each measurement packet. For that purpose, we utilize the physical NIC $NIC_2$ in the measurement VM and synchronize this NIC's clock to \lstinline|CLOCK_SYNCTIME| using \lstinline|phc2sys|. Now, we can take a hardware timestamp of an incoming measurement packet and a software timestamp using \lstinline|CLOCK_SYNCTIME|. The difference between the hardware and the software timestamp quantifies the measurement error $\gamma_{virt}$ precisely, allowing us to account for it in every measurement by subtracting $\gamma_{virt}$ from the measured clock synchronization precision $\Pi^*$. Note that this approach does not trim or transform the measured clock synchronization precision since any inaccuracy or degradation present in \lstinline|CLOCK_SYNCTIME| propagates to the NIC's clock and yields a poor clock synchronization precision measurement.

In summary, the clock synchronization precision $\Pi^*$ measured by the native nodes $mfn_1$ and $mfn_2$, or their respective measurement VMs in a virtualized setup, can deviate from the actual clock synchronization precision $\Pi$ by an inherent measurement error $\gamma$ so that it always is $\Pi^*\le \Pi + \gamma$ whereas $\Pi$ is the upper bound of clock synchronization precision derived from our model.

\subsection{Experimental Results}
For each experiment, to begin with, we substitute the minimum and maximum latencies in the formulas of the path message delay in order to derive the reading error and with it the upper bound on clock synchronization precision according to the formal model. We summarize all used minimum and maximum latency values in Table~\ref{table:latency_substitutions}.

\begin{description}[style=unboxed,leftmargin=0cm]
	\item[Native Setup with Hardware Timestamping] In case of the native setup with hardware timestamping, we can substitute the minimum and maximum transmission and reception timestamp latency for all nodes with the values found in the data sheet of the Intel I210 NIC~\cite{i210_data_sheet} yielding $\mathcal{L}_{min}^n(txts)=984\,ns$, $\mathcal{L}_{max}^n(txts)=1024\,ns$ and $\mathcal{L}_{min}^n(rxts)=2148\,ns$, $\mathcal{L}_{max}^n(rxts)=2228\,ns$ for $n\in\{sw_1,sw_2,sw_3,mfn_1,mfn_2,mfn_3\}$. We must note that the data sheet only provides values for $100\,Mb/s$ link speed thus we configured all NICs to operate at $100\,Mb/s$ to ensure validity of the provided timing characteristics. For the substitution of the minimum and maximum medium access time latency and the minimum and maximum transmission latency on the cable, we use the minimum and maximum propagation delays between adjacent nodes as measured by LinuxPTP for the execution of gPTP. For example, we measure a minimum and maximum latency of $\mathcal{L}_{min}^{gm,sw_1}(ma)+\mathcal{L}_{min}^{gm,sw_1}(c)=261\,ns$ and $\mathcal{L}_{max}^{gm,sw_1}(ma)+\mathcal{L}_{max}^{gm,sw_1}(c)=286\,ns$ between the grandmaster $gm$ and switch $sw_1$.
	\begin{figure}[t]
		\subfigure[Clock sync. precision, native HWTS]{
			\includegraphics[width=0.5\textwidth]{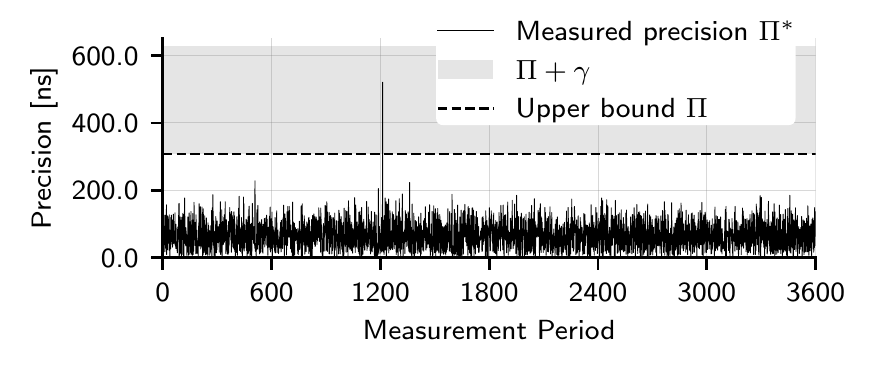}
			\label{fig:native_hwts_precision}
		}\hfill
		\subfigure[Distribution of values, native HWTS]{
			\includegraphics[width=0.5\textwidth]{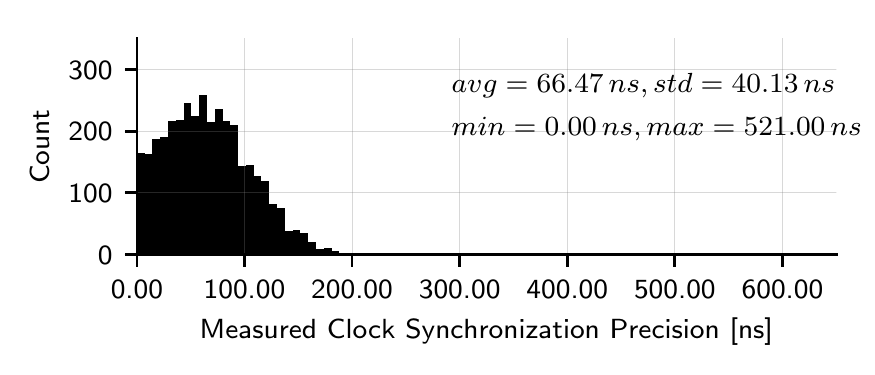}
			\label{fig:native_hwts_histogram}
		}
		\caption{Showing the upper bound of clock synchronization precision as derived from our formal model (dashed) and the clock synchronization precision measured in the native setup (black) with hardware timestamping (HWTS) in (a) and (b). For the measured clock synchronization precision in (a) we can see a single outlier that falls within the measurement error $\gamma=321\,ns$~(c.f.~Section~\ref{ssec:clock_sync_precision_measurement}).}
	\end{figure}
	After substituting all values, we find minimum and maximum path message delays between the grandmaster $gm$ and nodes $mfn_1$ and $mfn_2$ of:
	\begin{align}
	d_{min}^{gm,mfn_1}=6667\,ns \hspace{1cm}\text{ and }\hspace{1cm} d_{max}^{gm,mfn_1}=6972\,ns\nonumber\\
	d_{min}^{gm,mfn_2}=6672\,ns \hspace{1cm}\text{ and }\hspace{1cm} d_{max}^{gm,mfn_2}=6969\,ns
	\end{align}
	As a result, we obtain a reading error of $\mathcal{E}=6972\,ns-6667\,ns=305\,ns$. Together with the drift offset this yields an upper bound of clock synchronization precision of $\Pi=307.96\,ns$.
	
	Furthermore, using the values in Table~\ref{table:latency_substitutions}, we find an inherent measurement delay of $\gamma=321\,ns$ as discussed in Section~\ref{ssec:clock_sync_precision_measurement}. In Figure~\ref{fig:native_hwts_precision}, we plotted the measured clock synchronization precision for the native setup with hardware timestamping against the upper bound derived from our model. Furthermore, Figure~\ref{fig:native_hwts_histogram} illustrates the distribution of the measured clock synchronization precision showing an average of $66.47\,ns$ with a standard deviation of $40.13\,ns$ and a maximum of $521\,ns$.

	\item[ACRN with BVT Scheduling and Hardware Timestamping] In the resource consolidating hypervisor setup, we configured ACRN to schedule four VMs $vm_0$, $vm_1$, $vm_2$, and $vm_3$ according to BVT on cores $0$ to $2$ whereas $vm_4$, the measurement VM, is pinned to core~$3$. After the experiment ran for 30 minutes, we automatically started the execution of \lstinline|stress-ng|\footnote{\url{https://github.com/ColinIanKing/stress-ng}} with eight CPU workers each on VMs $vm_0$, $vm_2$, and $vm_3$ in order to test how processor load affects clock synchronization in $vm_1$ due to increased vCPU scheduling.
	\begin{figure}
		\subfigure[Clock sync. precision, ACRN BVT HWTS]{
			\includegraphics[width=0.5\textwidth]{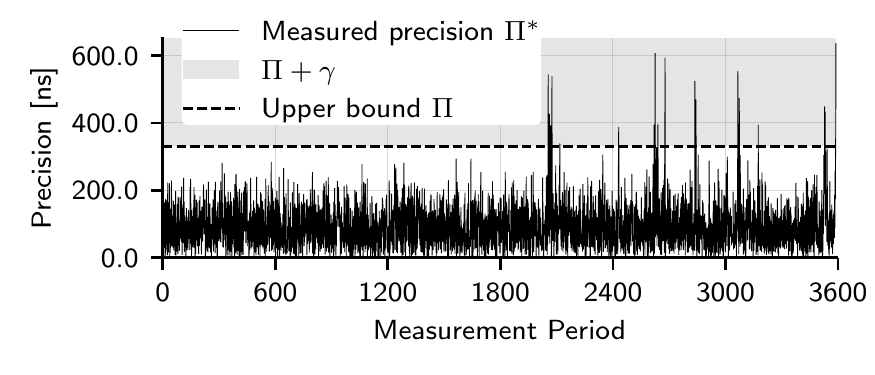}
			\label{fig:acrn_bvt_hwts_precision}
		}\hfill
		\subfigure[Distribution of values, ACRN BVT HWTS]{
			\includegraphics[width=0.5\textwidth]{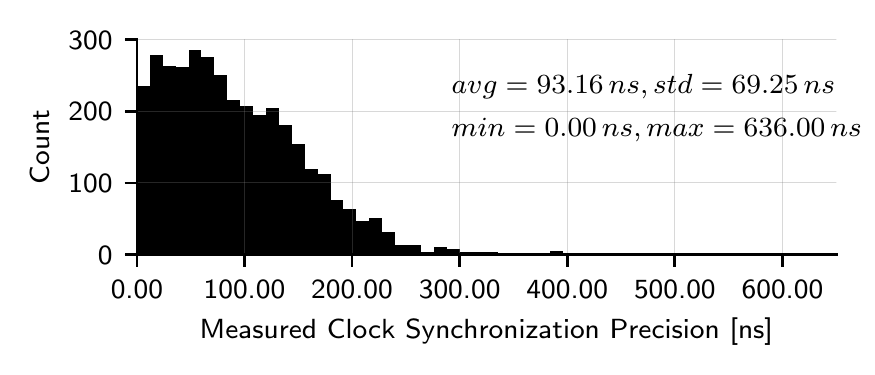}
			\label{fig:acrn_bvt_hwts_histogram}
		}
		\caption{Showing the upper bound of clock synchronization precision as derived from our formal model (dashed) and the clock synchronization precision measured in the ACRN BVT setup (black) with hardware timestamping (HWTS) in (a) and (b). The expected measurement error equals to $\gamma=321\,ns$~(c.f.~Section~\ref{ssec:clock_sync_precision_measurement}) the same as in the native setup before.}
	\end{figure}
	When using HWTS with the clock synchronization VM, we bypass all hypervisor latencies thus obtaining $\mathcal{L}_{min}(hv)=\mathcal{L}_{min}(vme)+\mathcal{L}_{min}(sched)=0\,s$ and $\mathcal{L}_{max}(hv)=\mathcal{L}_{max}(vme)+\mathcal{L}_{max}(sched)=0\,s$. This results in a reading error of $\mathcal{E}=327\,ns$ that only slightly differs from the reading error of the native setup with hardware timestamping due to the adapted network topology in the virtualized setups (c.f.~Figure~\ref{fig:experiment_configuration_baseline}~and~Figure~\ref{fig:experiment_configuration}). Therefore, we obtain an upper bound of clock synchronization precision of $\Pi=329.96\,ns$.
	
	In Figures~\ref{fig:acrn_bvt_hwts_precision}~and~\ref{fig:acrn_bvt_hwts_histogram}, we plotted the measured clock synchronization precision and its upper bound for ACRN in a resource consolidating configuration with hardware timestamping. For completeness, Figures~\ref{fig:bvt_mfn1_measurement_error}~and~\ref{fig:bvt_mfn2_measurement_error} show the measurement error for measurement VMs $mfn_1^{vm_4}$ and $mfn_2^{vm_4}$ due to hypervisor and guest OS latencies, which we have described in Section~\ref{ssec:clock_sync_precision_measurement}, and that have already been accounted for in Figures~\ref{fig:acrn_bvt_hwts_precision}~and~\ref{fig:acrn_bvt_hwts_histogram}. We note that for the first $1800$ measurement periods, the clock synchronization precision is constantly lower than the given upper bound. From measurement period $1800$ onward, with the activation of the \lstinline|stress-ng| CPU workers, we observe an increase of outliers that still lie within the expected measurement error.
	
	\item[ACRN with Pinned VCPUs and Hardware Timestamping] For the resource partitioning setup, we configured ACRN to pin VMs to processor cores so that $vm_0$ is pinned to core $0$, $vm_1$ to core $1$, and so forth. As before, after the experiment ran for 30 minutes, we automatically started the execution of \lstinline|stress-ng| with eight CPU workers each on VMs $vm_0$ and $vm_2$ in order to test how processor load affects clock synchronization in $vm_1$ in a partitioned setup.
	\begin{figure}[t]
		\subfigure[Clock sync. precision, ACRN pinned HWTS]{
			\includegraphics[width=0.5\textwidth]{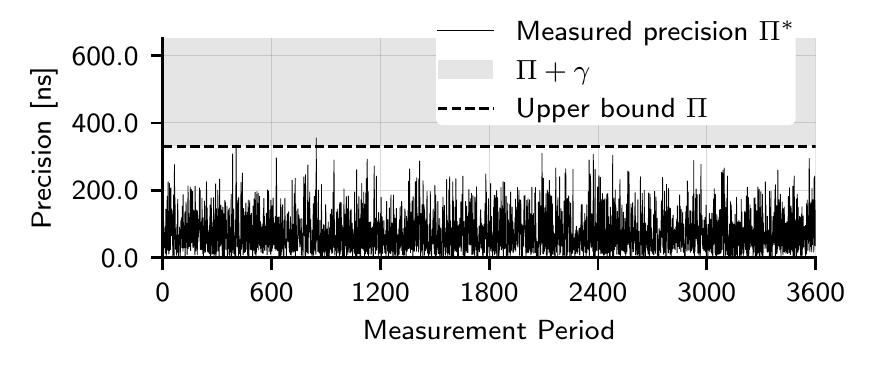}
			\label{fig:acrn_vmexitless_hwts_precision}
		}\hfill
		\subfigure[Distribution of values, ACRN pinned HWTS]{
			\includegraphics[width=0.5\textwidth]{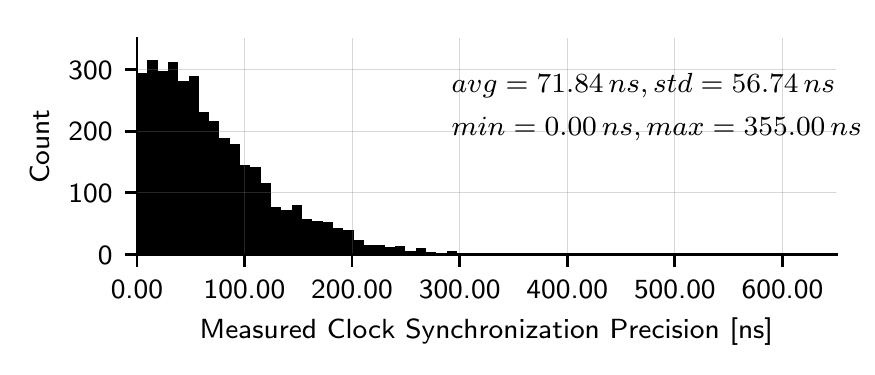}
			\label{fig:acrn_vmexitless_hwts_histogram}
		}
		\caption{Showing the upper bound of clock synchronization precision as derived from our formal model (dashed) and the clock synchronization precision measured in the ACRN VM exit-less setup (black) with hardware timestamping (HWTS) in (a) and (b). The expected measurement error equals to $\gamma=321\,ns$~(c.f.~Section~\ref{ssec:clock_sync_precision_measurement}) the same as in the previous setups.}
	\end{figure}
	We again bypass all hypervisor latencies by using a pass-through NIC and the dependent clock paradigm yielding a reading error of $\mathcal{E}=327\,ns$ that approximates the reading error of the native setup with hardware timestamping and accordingly also guarantees a precision $\Pi=329.96\,ns$ of the same magnitude. 
	
	Figures~\ref{fig:acrn_vmexitless_hwts_precision}~and~\ref{fig:acrn_vmexitless_hwts_histogram}, show the corresponding measured clock synchronization precision and the upper bound. We have again plotted the measurement error due to hypervisor and guest OS latencies for measurement VMs $mfn_1^{vm_3}$ and $mfn_2^{vm_3}$ in Figures~\ref{fig:vmexitless_mfn1_measurement_error}~and~\ref{fig:vmexitless_mfn2_measurement_error} respectively.
\end{description}

\subsection{Discussion}
Our experimental results show that it is feasible to achieve almost native clock synchronization precision with a type-1 hypervisor that implements the dependent clock paradigm. This is true for both a resource partitioning as well as a resource consolidating hypervisor configuration. However, we do note a slight degradation of clock synchronization precision in Figure~\ref{fig:acrn_bvt_hwts_precision} compared to Figure~\ref{fig:acrn_vmexitless_hwts_precision} when we synthetically create processor load using \lstinline|stress-ng| starting with measurement period $1800$. Even though the upper bound of clock synchronization is not violated outside the expected inherent measurement error $\gamma$, we have to investigate what causes this degradation and how we can cope with it. A possibility is that the performance of \lstinline|phc2sys|, which is synchronizing \lstinline|CLOCK_SYNCTIME| to the NIC's clock, is the source of degradation since our formal model does not capture it.

The upper bound of clock synchronization precision derived using our formalization holds if we consider the inherent measurement error $\gamma$. We can observe single values for all three experiments that supposedly violate the upper bound of clock synchronization precision. However, these observations always fall within the expected measurement error. We also want to note that there is still the possibility to account for the measurement error in the upper bound of clock synchronization precision, in which case the upper bound always holds. We have discussed in Section~\ref{ssec:clock_sync_precision_measurement} why we have disregarded this presentation and decided to plot measurement error and clock synchronization precision separately.

Finally, we admit that utilizing a physical NIC with hardware timestamping to bypass hypervisor and OS latency can be considered common practice so that the almost native clock synchronization performance seen in our experiments might not come as a surprise. However, we want to point out that there is a lack of publicly available research and prototypes that showcase this common practice. Furthermore, with our detailed analysis of clock synchronization in virtualized distributed real-time systems in Section~\ref{ssec:degradation_of_clock_sync_precision} we provide the theoretical framework that supports and predicts these practical results. Finally, identifying the measurement error due to hypervisor and guest OS latency of the measurement VMs as illustrated in Figure~\ref{fig:virt_measurement_error} that effectively reveals the full potential of the dependent clock paradigm in type-1 hypervisors is a direct result of the initial detailed analysis of the problem.

\section{Conclusion and Future Work}
The introduction of virtualization in time-triggered systems faces challenges when it comes to the integration of legacy applications and time-triggered communication and computation. A precondition for applying the time-triggered paradigm is the availability of a global time base that is established by the execution of a clock synchronization protocol. We introduce the notion of virtual clocks and provide a virtual clock condition that, if fulfilled, guarantees a virtual clock to be a good clock and thus is suited for clock synchronization in a virtualized distributed real-time system. After showing that discontinuous virtual clocks fulfill the virtual clock condition, we investigate how virtualization degrades clock synchronization in distributed real-time systems and quantify the relevant parameters. We use this formalization to derive an upper bound on clock synchronization precision of IEEE\,802.1AS in a virtualized distributed real-time system.

Our experimental results show that near-native clock synchronization precision using a type-1 hypervisor that implements the dependent clock paradigm is possible given that we account for the measurement error due to hypervisor and guest OS latency in the measurement VMs. To the authors' knowledge, this is the first prototype of a type-1 hypervisor that achieves sub-microsecond clock synchronization precision and additionally comes with a formalization supporting this claim. We see this as proof that the dependent clock paradigm can drive the scheduler of a time-triggered hypervisor in a virtualized distributed real-time system, enabling the virtualization of distributed control in many use cases and industries.

In future work, we will extend our formalization and implementation by fault tolerance measures and integrate a time-triggered VCPU scheduler that uses the synchronized time to schedule VCPUs according to a global task and network schedule.

\section*{Acknowledgment}
The research leading to these results has received funding from the European Union's Horizon 2020 research and innovation programme under the Marie Sklodowska-Curie grant agreement No.~764785, FORA -- Fog Computing for Robotics and Industrial Automation

\newpage

\appendix
\section{Appendix}
\subsection{Notation Summary}
\label{app:notation_summary}
\begin{table}[h!]
	\centering
	\caption{Summary of the notation used throughout the paper.}
	\begin{tabularx}{\textwidth}{cc|l}
		& Notation & Description \\
		\hline\hline
		\multirow{7}{*}{\rotatebox[origin=c]{90}{Section~\ref{sec:notion_of_time}}} & $ts(e)$ & Timestamp of event $e$ in the reference time\\
		& $tl^k(e)$ & Timestamp of event $e$ in the local time $tl^k$ of node $k$ \\
		& $l$ & Microtick of a local physical clock \\
		& $C^k$ & Local physical clock of node $k$ \\
		& $g^k$ & Granularity of the clock $C^k$ of node $k$ \\
		& $r^k_{l,l+1}$ & Clock drift of node $k$ between two consecutive microtick $l$ and $l+1$ \\
		& $r_{max}$ & Maximum drift rate\\
		\hline
		
		\multirow{4}{*}{\rotatebox[origin=c]{90}{Section~\ref{ssec:virtual_clocks}}} & $\mathcal{V}^k$ & Set of VMs hosted on physical node $k$ \\
		& $C^{k,j}$ & Local virtual clock of VM $j$ running on physical node $k$\\
		& $tv^{k,j}(e)$ & Timestamp of event $e$ in the virtual time $tv^{k,j}$ of VM $j$ hosted on node $k$\\
		& $\nu^{k,j}$ & Virtual microtick of VM $j$ hosted on physical node $k$\\
		\hline
		
		\multirow{4}{*}{\rotatebox[origin=c]{90}{Section~\ref{ssec:continuous_discontinuous_clocks}}} & $o(l)$ & The clock offset of a virtual clock from its local clock at microtick $l$\\
		& $\delta$ & Generic duration of the preemption of a VM by the hypervisor\\
		& $P_l$ & Set of the durations of all preemptions of a VM up to microtick $l$ of the host\\
		& $p_i\in P_l$ & Duration of the $i$-th preemption of a VM\\
		\hline
		
		\multirow{3}{*}{\rotatebox[origin=c]{90}{Section~\ref{ssec:clock_sync_precision}}} & $D$ & Union set of all nodes, virtual and physical, of a virtualized distributed system\\
		& $G$ & Set of all good clocks in a virtualized distributed system $D$ \\
		& $\Pi$ & Precision of clock synchronization (c.f.~Definition~\ref{def:clock_synchronization_precision}) \\
		\hline\hline

		\multirow{5}{*}{\rotatebox[origin=c]{90}{Section~\ref{sec:clock_sync_in_virtualized_distributed_rt_systems}}} & $S$ & Resynchronization period of clock synchronization\\
		& $s\in\mathbb{N}_0$ & Specific resynchronization period\\
		& $m_s^{n,n'}$ & Clock synchronization message from node $n$ to $n'$ in period $s$\\
		& $tx^{n'}_s$ & Transmit event of a synchronization message in period $s$ to receiver $n'\in D$\\
		& $rx^{n}_s$ & Reception event of a synchronization message in period $s$ from sender $n\in D$\\
		\hline
		
		\multirow{8}{*}{\rotatebox[origin=c]{90}{Sections~\ref{ssec:reading_error}~and~\ref{ssec:synchronization_condition}}} & $c_s^{n,n'}$ & The correction term between node $n$ and $n'$ (c.f.~Definition~\ref{def:correction_term})\\
		& $e_s^{n,n'}$ & Estimated message delay between nodes $n$ and $n'$ during period $s$ \\
		& $a_s^{n,n'}$ & Actual message delay between nodes $n$ and $n'$ during period $s$ \\
		& $\Delta_s^{n,n'}$ & Actual time difference between clock $C^n$ and clock $C^{n'}$\\
		& $\epsilon_s^{n,n'}$ & Reading delay between nodes $n$ and $n'$ during period $s$\\
		& $d_{min}$ & Minimum message delay in a virtualized distributed system $D$\\
		& $d_{max}$ & Maximum message delay in a virtualized distributed system $D$\\
		& $\mathcal{E}$ & Reading error in a virtualized distributed system $D$ (c.f.~Definition~\ref{def:reading_error})\\
		& $\Gamma$ & Drift offset in a virtualized distributed system $D$\\
		
		\hline
		\multirow{7}{*}{\rotatebox[origin=c]{90}{Section~\ref{ssec:degradation_of_clock_sync_precision}}} & $d_{min}^{k,k'}$ & Minimum (path) message delay between nodes $k$ and $k'$\\
		& $d_{max}^{k,k'}$ & Maximum (path) message delay between nodes $k$ and $k'$\\
		& $\mathcal{L}_{min}^k(x)$ & Minimum latency introduced on node $k$ by component $x$\\
		& $\mathcal{L}_{max}^k(x)$ & Maximum latency introduced on node $k$ by component $x$\\
		& $E_{n,n'}$ & Path from a node $n$ to a node $n'$\\
		& $(k,k')\in E_{n,n'}$ & Link between adjacent nodes $k,k'\in D$ on path $E_{n,n'}$\\
		& $E_D$ & Set of all possible paths between two nodes in a virt. distributed system $D$\\
	\end{tabularx}
\end{table}

\bibliographystyle{abbrvnat}
\bibliography{references} 
\end{document}